\documentclass[submission,copyright,creativecommons]{eptcs}
\providecommand{\event}{QPL 2017}
\pdfoutput=1
\usepackage{amsthm}
\usepackage{amssymb}
\usepackage{amsfonts}
\usepackage{amsmath}
\usepackage{stmaryrd}
\usepackage{xspace}
\usepackage{mathtools}
\usepackage{multicol}

\theoremstyle{plain}
\newtheorem{theorem}{Theorem}
\newtheorem{proposition}[theorem]{Proposition}
\newtheorem*{example}{Example}
\newtheorem{definition}[theorem]{Definition}
\newtheorem*{remark}{Remark}
\newtheorem{corollary}[theorem]{Corollary}
\newtheorem{lemma}[theorem]{Lemma}

\allowdisplaybreaks 

\usepackage{tikz}
\usepackage{pgfplots}
\usetikzlibrary{decorations.markings}
\usetikzlibrary{matrix,backgrounds,folding}

\pgfdeclarelayer{edgelayer}
\pgfdeclarelayer{nodelayer}
\pgfsetlayers{background,edgelayer,nodelayer,main}
\tikzstyle{none}=[inner sep=0mm]
\tikzstyle{every picture}=[baseline=-0.25em]
\tikzstyle{black dot}=[inner sep=0.7mm,minimum width=0pt,minimum height=0pt,fill=black,draw=black,shape=circle]
\tikzstyle{dot}=[black dot]
\tikzstyle{zxnode}=[shape=circle, minimum width=.25cm, inner sep=0.5pt, font=\footnotesize, draw=black]
\tikzstyle{gn}=[zxnode ,fill=green]
\tikzstyle{rn}=[zxnode ,fill=red]
\tikzstyle{H box}=[rectangle,fill=yellow,draw=black,xscale=1,yscale=1,font=\small,inner sep=0.75pt,minimum width=0.15cm,minimum height=0.15cm]\tikzstyle{bx}=[rectangle, fill=white, 
bottom color=red, minimum width=0.4cm, minimum height=0.4cm, top color=green, middle color=white,
 draw=black, font=\footnotesize ,allow upside down, inner sep=0.15em]
\tikzstyle{udbx}=[rectangle, fill=white, bottom color=green, minimum width=0.4cm, minimum height=0.4cm, top color=red, middle color=white, draw=black, font=\footnotesize ,allow upside down, inner sep=0.15em]
\tikzstyle{rbx}=[rectangle, fill=white, right color=red, minimum width=0.4cm, minimum height=0.4cm, left color=green, middle color=white, draw=black, font=\footnotesize ,allow upside down, inner sep=0.15em]
\tikzstyle{lbx}=[rectangle, fill=white, left color=red, minimum width=0.4cm, minimum height=0.4cm, right color=green, middle color=white, draw=black, font=\footnotesize ,allow upside down, inner sep=0.15em]
\tikzstyle{diredge}=[->]
\tikzstyle{every loop}=[]

\newcommand{\interp}[1] {\left\llbracket #1 \right\rrbracket}
\newcommand{\fit}[1] {\resizebox{\columnwidth}{!}{#1}}

\newcommand{\redstate}{\begin{pmatrix}1 \\ 0\end{pmatrix}}
\newcommand{\tredstate}{\begin{pmatrix}1 & 0\end{pmatrix}}

\newcommand{\frag}[1]{$\frac{\pi}{#1}$-fragment}
\newcommand{\norm}[1]{\lVert #1 \rVert}
\renewcommand{\implies}{\quad\Rightarrow\quad}
\renewcommand{\mod}{\bmod}
\renewcommand{\Re}{\operatorname{Re}}
\renewcommand{\Im}{\operatorname{Im}}

\newcommand{\annoted}[3]{{\scriptstyle #1}\left\lbrace\mathrlap{\phantom{#3}}\right.\overbrace{#3}^{#2}}

\newcommand{\iY}{\begin{pmatrix}0 & 1 \\ -1 & 0\end{pmatrix}}
\newcommand{\nmswapI}[2]{
\begin{tikzpicture}
	\begin{pgfonlayer}{nodelayer}
		\node [style=none] (0) at (-0.7500001, -0) {};
		\node [style=none] (1) at (0.2499998, -0) {};
		\node [style=none] (2) at (-0.7500001, -0.7500001) {};
		\node [style=none] (3) at (0.2499998, -0.7500001) {};
		\node [style=none] (4) at (0.7500001, -0) {};
		\node [style=none] (5) at (-0.2499998, -0.7500001) {};
		\node [style=none] (6) at (-0.2499998, -0) {};
		\node [style=none] (7) at (0.7500001, -0.7500001) {};
		\node [style=none] (8) at (-0.5000002, -0) {$\cdots$};
		\node [style=none] (9) at (0.5000002, -0) {$\cdots$};
		\node [style=none] (10) at (-0.5000002, -0.7500001) {$\cdots$};
		\node [style=none] (11) at (0.5000002, -0.7500001) {$\cdots$};
		\node [style=none] (12) at (-0.875, 0.2500001) {};
		\node [style=none] (13) at (-0.125, 0.2500001) {};
		\node [style=none] (14) at (-0.5000002, 0.5) {};
		\node [style=none] (15) at (-0.5000002, 0.7500001) {$#1$};
		\node [style=none] (16) at (0.125, 0.2500001) {};
		\node [style=none] (17) at (0.875, 0.2500001) {};
		\node [style=none] (18) at (0.5000002, 0.7500001) {$#2$};
		\node [style=none] (19) at (0.5000002, 0.5) {};
		\node [style=none] (20) at (0.9999999, -0) {};
		\node [style=none] (21) at (0.9999999, -0.7500001) {};
		\node [style=none] (22) at (1.25, -0.2500001) {};
	\end{pgfonlayer}
	\begin{pgfonlayer}{edgelayer}
		\draw [in=90, out=-90, looseness=0.75] (0.center) to (3.center);
		\draw [in=90, out=-90, looseness=0.75] (1.center) to (2.center);
		\draw [in=90, out=-90, looseness=0.75] (6.center) to (7.center);
		\draw [in=90, out=-90, looseness=0.75] (4.center) to (5.center);
		\draw [in=-90, out=90, looseness=1.50] (12.center) to (14.center);
		\draw [in=90, out=-90, looseness=1.50] (14.center) to (13.center);
		\draw [in=-90, out=90, looseness=1.50] (16.center) to (19.center);
		\draw [in=90, out=-90, looseness=1.50] (19.center) to (17.center);
		\draw (20.center) to (21.center);
	\end{pgfonlayer}
\end{tikzpicture}}

\newcommand{\hlabel}{\phantomsection\label}

\newcommand{\callrule}[2]{\hyperlink{r:#1}{\textnormal{(#2)}}\xspace}

\newcommand{\so}{\callrule{rules}{S1}}
\newcommand{\st}{\callrule{rules}{S2}}
\newcommand{\sth}{\callrule{rules}{S3}}
\newcommand{\iv}{\callrule{rules}{IV}}
\newcommand{\bo}{\callrule{rules}{B1}}
\newcommand{\bt}{\callrule{rules}{B2}}
\newcommand{\rso}{\callrule{rules}{RS1}}
\newcommand{\rsoo}{\ref{eq:real-spider}\xspace}

\newcommand{\rsth}{\callrule{rules}{RS2}}
\newcommand{\rh}{\callrule{rules}{RH}}
\newcommand{\rzo}{\callrule{rules}{RZO}}
\newcommand{\rsup}{\callrule{rules}{RSUP}}
\newcommand{\rsupc}[1]{\callrule{rules}{RSUP}}
\newcommand{\zxrso}{\callrule{zxr-rules}{S1}}
\newcommand{\zxrst}{\callrule{zxr-rules}{S2}}
\newcommand{\zxrsth}{\callrule{zxr-rules}{S3}}
\newcommand{\zxriv}{\callrule{zxr-rules}{IV}}
\newcommand{\zxrbo}{\callrule{zxr-rules}{B1}}
\newcommand{\zxrbt}{\callrule{zxr-rules}{B2}}
\newcommand{\zxrhl}{\callrule{zxr-rules}{HL}}
\newcommand{\zxrh}{\callrule{zxr-rules}{H}}
\newcommand{\zxrzo}{\callrule{zxr-rules}{ZO}}

\newcommand{\zxso}{\callrule{zx-rules}{S1}}
\newcommand{\zxst}{\callrule{zx-rules}{S2}}
\newcommand{\zxsth}{\callrule{zx-rules}{S3}}
\newcommand{\zxe}{\callrule{zx-rules}{E}}
\newcommand{\zxbo}{\callrule{zx-rules}{B1}}
\newcommand{\zxbt}{\callrule{zx-rules}{B2}}
\newcommand{\zxeu}{\callrule{zx-rules}{EU}}
\newcommand{\zxh}{\callrule{zx-rules}{H}}
\newcommand{\zxkt}{\callrule{zx-rules}{K2}}
\newcommand{\zxsup}{\callrule{zx-rules}{SUP}}
\newcommand{\zxsupc}[1]{\callrule{zx-rules}{SUP}}

\newcommand{\zxzo}{\callrule{zx2}{ZO}}
\newcommand{\zxiv}{\callrule{zx2}{IV}}

\def \reca {Y-Calculus\xspace}
\def \reda {Y-diagram\xspace}
\def \redas {Y-diagrams\xspace}
\def \rec {Y\xspace}
\def \repiot {\textnormal{Y$_{\frac{\pi}{2}}$}\xspace}
\newcommand{\reind}[1]{\textnormal{Y$_{#1}$}\xspace}

\title{\reca: A Language for Real Matrices Derived from the ZX-Calculus}

\author{
Emmanuel Jeandel $\qquad\qquad$ Simon Perdrix $\qquad\qquad$ Renaud Vilmart
\institute{LORIA, CNRS, Universit\'e de Lorraine, Inria, F 54000 Nancy, France}
\email{emmanuel.jeandel@loria.fr \qquad\qquad simon.perdrix@loria.fr \qquad\qquad renaud.vilmart@loria.fr}
}
\def\titlerunning{\reca}
\def\authorrunning{E. Jeandel, S. Perdrix \& R. Vilmart}

\begin{document}
\maketitle
\begin{abstract}
We introduce a ZX-like diagrammatic language devoted to manipulating real matrices -- and rebits --, with its own set of axioms. We prove the necessity of some non trivial axioms of these. We show that some restriction of the language is complete. 
We exhibit two interpretations to and from the ZX-Calculus, thus showing the consistency between the two languages. Finally, we derive from our work a way to extract the real or imaginary part of a ZX-diagram, and prove that a restriction of our language is complete if the equivalent restriction of the ZX-calculus is complete.
\end{abstract}

\section{Introduction}

The ZX-Calculus, introduced by Coecke and Duncan \cite{interacting}, is a powerful formal tensor-language for quantum reasoning \cite{coecke2017picturing}. 
The ZX-calculus is based upon the axiomatisation of interacting observables (Pauli-X and Pauli-Z) together with rotations around X- and Z-axis. Both X- and Z-observables are real but X- and Z-rotations are not.  This is a universal language for quantum mechanics: any complex $2^n\times 2^m$-matrix can be represented. The ZX-calculus can be used to represent quantum circuits as well as measurement-based quantum computations \cite{mbqc,CP-state-transfer,CP12}. The angle-free version of the ZX-calculus has been proved to be universal for real stabilizer quantum mechanics \cite{pivoting}, a non universal fragment of quantum mechanics.

In this paper we introduce a ZX-like language for real matrices, called Y-calculus. The introduction of the Y-calculus has multiple motivations: 
\begin{itemize}
\item[($i$)] First, diagrammatic languages, like the ZX-calculus, are not necessarily devoted to quantum applications \cite{Bonchi2014,CSC10}, and dealing with real matrices might be more convenient than complex matrices. 
\item[($ii$)] Moreover, real quantum mechanics is a sub-quantum theory of interests, from the very foundational questions  to quantum information processing: the use of real rather than complex numbers in quantum mechanics is related to local tomography \cite{locality}; As a model of computation using real instead of complex numbers does not change its computational power \cite{BernsteinVazirani,MP12}. Moreover, real quantum computation is used e.g.~for interactive proofs \cite{complexity,perdrix:hal-01377339} or to study contextuality \cite{context-rebits}. 
\item[($iii$)] The axiomatisation of two interacting observables is the cornerstone of the ZX-calculus. These two observables correspond to the so-called two unbiased basis of rebits (real qubits). The ZX-calculus fails to capture in a simple way the third unbiased basis \cite{trichromatic} which occurs only in the complex case. As a consequence, the angle-free ZX-calculus seems to be better suited for real quantum mechanics than complex quantum mechanics. We explore this line of research in the present paper by equipping  the angle-free ZX-calculus with real rotations. 
\end{itemize}

The Y-calculus is based on the same complementary observables (Pauli-X and Pauli-Z) as the ZX-calculus. To make it universal for real quantum mechanics we axiomatise the Y-rotations which are real rotations. Notice that Y-rotations have been axiomatised by Lang and Coecke  \cite{trichromatic}, however they use non-real matrices to represent Y-rotations, and they axiomatise Y-rotations together with the X- and Z-rotations, the interactions of the three rotations leading to a combinatorial explosion of the rules of the language which is avoided in the Y-calculus which only deals with Y-rotations.

One of the main open question for tensor-like languages like the ZX-calculus is the completeness of the language. The language would be complete if, for any two diagrams that represent the same matrix, they could be transformed into one-another only using the transformation rules allowed by the language. The ZX-Calculus is not complete in general \cite{incompleteness}, but some of its fragments are. The $\pi$-fragment and the \frag{2} are both complete \cite{pivoting,pi_2-complete}. The \frag{4}, unlike the $\pi$- and the \frag{2}, is \emph{approximately universal} \cite{clifford+t}, meaning that any quantum evolution can be approximated with arbitrarily good precision with this fragment. Notice that a complete axiomatisation for the  \frag{4} has been recently introduced \cite{complete}.

In section \ref{sec:calculus}, we present the ZX-Calculus and define the \reca. We give a set of rules to this language, and prove that two of its non-trivial axioms are not derivable from the others (section \ref{sec:minimality}). We establish a link between the \frag{2} of the \reca and the $\pi$-fragment of the ZX-Calculus, and thanks to the completeness of the latter, we prove the \frag{2} of the \reca is complete (section \ref{sec:pi_2-completeness}). 
We finally exhibit an interpretation from the \reca to the ZX-Calculus (section \ref{sec:y-to-zx}), which shows the consistency of the two languages, and another interpretation from the ZX-Calculus to the \reca 
which show they have the same power: ZX-calculus is complete if and only if Y-calculus is complete.

\section{ZX and \rec-Calculi}
\label{sec:calculus}

\subsection{ZX-Calculus}
\label{sec:zx-def}

A ZX-diagram $D:k\to l$ is an open diagram with $k$ inputs and $l$ outputs and is generated by:
\begin{center}
\bgroup
\def\arraystretch{2.5}
\begin{tabular}{|cc|cc|}
\hline
$R_Z^{(n,m)}(\alpha):n\to m$ & \begin{tikzpicture}
	\begin{pgfonlayer}{nodelayer}
		\node [style=gn] (0) at (0, -0) {$~\alpha~$};
		\node [style=none] (1) at (-0.5000001, 0.7499999) {};
		\node [style=none] (2) at (0, 0.4999999) {$\cdots$};
		\node [style=none] (3) at (0.5000001, 0.7499999) {};
		\node [style=none] (4) at (-0.5000001, -0.7499999) {};
		\node [style=none] (5) at (0, -0.4999999) {$\cdots$};
		\node [style=none] (6) at (0.5000001, -0.7499999) {};
		\node [style=none] (7) at (0, 0.75) {$n$};
		\node [style=none] (8) at (0, -0.7499998) {$m$};
	\end{pgfonlayer}
	\begin{pgfonlayer}{edgelayer}
		\draw [style=none, bend right, looseness=1.00] (1.center) to (0);
		\draw [style=none, bend right, looseness=1.00] (0) to (3.center);
		\draw [style=none, bend right, looseness=1.00] (0) to (4.center);
		\draw [style=none, bend left, looseness=1.00] (0) to (6.center);
	\end{pgfonlayer}
\end{tikzpicture} & $R_X^{(n,m)}(\alpha):n\to m$ & \begin{tikzpicture}
	\begin{pgfonlayer}{nodelayer}
		\node [style=rn] (0) at (0, -0) {$~\alpha~$};
		\node [style=none] (1) at (-0.5000001, 0.7499999) {};
		\node [style=none] (2) at (0, 0.4999999) {$\cdots$};
		\node [style=none] (3) at (0.5000001, 0.7499999) {};
		\node [style=none] (4) at (-0.5000001, -0.7499999) {};
		\node [style=none] (5) at (0, -0.4999999) {$\cdots$};
		\node [style=none] (6) at (0.5000001, -0.7499999) {};
		\node [style=none] (7) at (0, 0.75) {$n$};
		\node [style=none] (8) at (0, -0.7499999) {$m$};
		\node [style=none] (9) at (0, -0.9999999) {};
		\node [style=none] (10) at (0, 0.9999999) {};
	\end{pgfonlayer}
	\begin{pgfonlayer}{edgelayer}
		\draw [style=none, bend right, looseness=1.00] (1.center) to (0);
		\draw [style=none, bend right, looseness=1.00] (0) to (3.center);
		\draw [style=none, bend right, looseness=1.00] (0) to (4.center);
		\draw [style=none, bend left, looseness=1.00] (0) to (6.center);
	\end{pgfonlayer}
\end{tikzpicture}\\[4ex]\hline
$H:1\to 1$ & \begin{tikzpicture}
	\begin{pgfonlayer}{nodelayer}
		\node [style={H box}] (0) at (0, 0) {};
		\node [style=none] (1) at (0, 0.5) {};
		\node [style=none] (2) at (0, -0.5) {};
	\end{pgfonlayer}
	\begin{pgfonlayer}{edgelayer}
		\draw (2.center) to (1.center);
	\end{pgfonlayer}
\end{tikzpicture}
 & $e:0\to 0$ & \begin{tikzpicture}
	\begin{pgfonlayer}{nodelayer}
		\node [style=none] (0) at (-0.2499999, 0.2499999) {};
		\node [style=none] (1) at (-0.2499999, -0.2499999) {};
		\node [style=none] (2) at (0.2499999, 0.2499999) {};
		\node [style=none] (3) at (0.2499999, -0.2499999) {};
	\end{pgfonlayer}
	\begin{pgfonlayer}{edgelayer}
		\draw [style=dashed] (0.center) to (2.center);
		\draw [style=dashed] (2.center) to (3.center);
		\draw [style=dashed] (3.center) to (1.center);
		\draw [style=dashed] (0.center) to (1.center);
	\end{pgfonlayer}
\end{tikzpicture}\\\hline
$\mathbb{I}:1\to 1$ & \begin{tikzpicture}
	\begin{pgfonlayer}{nodelayer}
		\node [style=none] (0) at (0, 0.2499999) {};
		\node [style=none] (1) at (0, -0.2499999) {};
	\end{pgfonlayer}
	\begin{pgfonlayer}{edgelayer}
		\draw (0.center) to (1.center);
	\end{pgfonlayer}
\end{tikzpicture} & $\sigma:2\to 2$ & \begin{tikzpicture}
	\begin{pgfonlayer}{nodelayer}
		\node [style=none] (0) at (-0.2499999, 0.2499999) {};
		\node [style=none] (1) at (0.2499999, -0.2499999) {};
		\node [style=none] (2) at (0.2499999, 0.2499999) {};
		\node [style=none] (3) at (-0.2499999, -0.2499999) {};
	\end{pgfonlayer}
	\begin{pgfonlayer}{edgelayer}
		\draw [in=90, out=-90, looseness=1.00] (0.center) to (1.center);
		\draw [in=90, out=-90, looseness=1.00] (2.center) to (3.center);
	\end{pgfonlayer}
\end{tikzpicture}\\\hline
$\epsilon:2\to 0$ & \begin{tikzpicture}
	\begin{pgfonlayer}{nodelayer}
		\node [style=none] (0) at (-0.2499999, 0.2499999) {};
		\node [style=none] (1) at (0.2499999, 0.2499999) {};
	\end{pgfonlayer}
	\begin{pgfonlayer}{edgelayer}
		\draw [in=-90, out=-90, looseness=2.00] (0.center) to (1.center);
	\end{pgfonlayer}
\end{tikzpicture} & $\eta:0\to 2$ & \raisebox{0.6em}{\begin{tikzpicture}
	\begin{pgfonlayer}{nodelayer}
		\node [style=none] (0) at (0.2499999, -0.2499999) {};
		\node [style=none] (1) at (-0.2499999, -0.2499999) {};
	\end{pgfonlayer}
	\begin{pgfonlayer}{edgelayer}
		\draw [in=90, out=90, looseness=2.00] (1.center) to (0.center);
	\end{pgfonlayer}
\end{tikzpicture}}\\\hline
\end{tabular}
\egroup\\
where $n,m\in \mathbb{N}$ and $\alpha \in \mathbb{R}$
\end{center}
and the two compositions:
\begin{itemize}
\item Spatial Composition: for any $D_1:a\to b$ and $D_2:c\to d$, $D_1\otimes D_2:a+c\to b+d$ consists in placing $D_1$ and $D_2$ side by side, $D_2$ on the right of $D_1$.
\item Sequential Composition: for any $D_1:a\to b$ and $D_2:b\to c$, $D_2\circ D_1:a\to c$ consists in placing $D_1$ on the top of $D_2$, connecting the outputs of $D_1$ to the inputs of $D_2$.
\end{itemize}

\begin{figure}[!htb]
 \centering
 \hypertarget{r:zx-rules}{}
  \begin{tabular}{|ccccc|}
   \hline
   &&&& \\
   \begin{tikzpicture}[font={\footnotesize}]
	\begin{pgfonlayer}{nodelayer}
		\node [style=none] (0) at (-1.25, -0) {\rotatebox[origin=c]{63.43}{$~\cdots~$}};
		\node [style=none] (1) at (0, -0) {$=$};
		\node [style=gn] (2) at (0.9999999, -0) { \footnotesize$\alpha{+}\beta$};
		\node [style=gn, minimum width={0.5 cm}] (3) at (-0.7500001, -0.2500001) {\footnotesize $\beta$};
		\node [style=none] (4) at (-1.75, -0.5) {$~\cdots~$};
		\node [style=none] (5) at (1.5, -0.7499999) {};
		\node [style=none] (6) at (-1, -0.75) {};
		\node [style=none] (7) at (0.9999999, -0.7499999) {$~\cdots~$};
		\node [style=none] (8) at (-0.5, -0.75) {};
		\node [style=none] (9) at (0.5000002, -0.7499999) {};
		\node [style=none] (10) at (-2, -0.5) {};
		\node [style=none] (11) at (-1.5, -0.5) {};
		\node [style=none] (12) at (-0.75, -0.75) {$~\cdots~$};
		\node [style=none] (13) at (0.9999999, 0.7499999) {$~\cdots~$};
		\node [style=none] (14) at (0.5000002, 0.7499999) {};
		\node [style=none] (15) at (-2, 0.75) {};
		\node [style=none] (16) at (-0.5, 0.5) {};
		\node [style=none] (17) at (-1.5, 0.75) {};
		\node [style=none] (18) at (1.5, 0.7499999) {};
		\node [style=gn, minimum width={0.5 cm}] (19) at (-1.75, 0.25) {\footnotesize$\alpha$};
		\node [style=none] (20) at (-0.75, 0.5) {$~\cdots~$};
		\node [style=none] (21) at (-1, 0.5) {};
		\node [style=none] (22) at (-1.75, 0.75) {$~\cdots~$};
	\end{pgfonlayer}
	\begin{pgfonlayer}{edgelayer}
		\draw (3) to (16.center);
		\draw (3) to (6.center);
		\draw (3) to (8.center);
		\draw (19) to (10.center);
		\draw (19) to (11.center);
		\draw [bend right, looseness=1.00] (19) to (3);
		\draw [bend left, looseness=1.00] (19) to (3);
		\draw (14.center) to (2);
		\draw (2) to (9.center);
		\draw (5.center) to (2);
		\draw (2) to (18.center);
		\draw (19) to (15.center);
		\draw (19) to (17.center);
		\draw (3) to (21.center);
	\end{pgfonlayer}
\end{tikzpicture}&(S1) &$\qquad$& \begin{tikzpicture}
	\begin{pgfonlayer}{nodelayer}
		\node [style=gn] (0) at (-0.7499998, -0) {};
		\node [style=none] (1) at (0, -0) {=};
		\node [style=none] (2) at (-0.7499998, 0.5) {};
		\node [style=none] (3) at (-0.7499998, -0.5000001) {};
		\node [style=none] (4) at (0.7499998, 0.5) {};
		\node [style=none] (5) at (0.7499998, -0.5000001) {};
		\node [style=none] (6) at (0, -0.7499998) {};
		\node [style=none] (7) at (0, 0.75) {};
	\end{pgfonlayer}
	\begin{pgfonlayer}{edgelayer}
		\draw (2.center) to (0);
		\draw (0) to (3.center);
		\draw (4.center) to (5.center);
	\end{pgfonlayer}
\end{tikzpicture}&(S2)\\
   &&&& \\
   \begin{tikzpicture}
	\begin{pgfonlayer}{nodelayer}
		\node [style=none] (0) at (0.7499998, -0.2500001) {};
		\node [style=none] (1) at (0, -0) {$=$};
		\node [style=gn] (2) at (1.25, 0.2500001) {};
		\node [style=none] (3) at (-0.7499998, -0.2500001) {};
		\node [style=none] (4) at (1.75, -0.2500001) {};
		\node [style=none] (5) at (-1.75, -0.2500001) {};
		\node [style=none] (6) at (0, 0.5) {};
		\node [style=none] (7) at (0, -0.5000001) {};
	\end{pgfonlayer}
	\begin{pgfonlayer}{edgelayer}
		\draw [in=90, out=90, looseness=1.75] (5.center) to (3.center);
		\draw [in=90, out=90, looseness=1.75] (0.center) to (4.center);
	\end{pgfonlayer}
\end{tikzpicture}&(S3) && \begin{tikzpicture}
	\begin{pgfonlayer}{nodelayer}
		\node [style=rn] (0) at (-0.7500001, -0.375) {$\frac{-\pi}{4}$};
		\node [style=gn] (1) at (-0.7499998, 0.375) {$~\frac{\pi}{4}~$};
		\node [style=none] (2) at (0, -0) {=};
		\node [style=none] (3) at (0.7500001, 0.25) {};
		\node [style=none] (4) at (0.7500001, -0.25) {};
		\node [style=none] (5) at (1.25, 0.25) {};
		\node [style=none] (6) at (1.25, -0.25) {};
	\end{pgfonlayer}
	\begin{pgfonlayer}{edgelayer}
		\draw (0) to (1);
		\draw [style=dashed] (3.center) to (5.center);
		\draw [style=dashed] (5.center) to (6.center);
		\draw [style=dashed] (6.center) to (4.center);
		\draw [style=dashed] (4.center) to (3.center);
	\end{pgfonlayer}
\end{tikzpicture}&(E)\\
   &&&& \\
   \begin{tikzpicture}
	\begin{pgfonlayer}{nodelayer}
		\node [style=gn] (0) at (0.75, 0) {};
		\node [style=none] (1) at (2.25, -0.25) {};
		\node [style=none] (2) at (0.5, -0.5) {};
		\node [style=rn] (3) at (2.25, 0.25) {};
		\node [style=none] (4) at (1, -0.5) {};
		\node [style=rn] (5) at (0.75, 0.5) {};
		\node [style=rn] (6) at (2.75, 0.25) {};
		\node [style=none] (7) at (2.75, -0.25) {};
		\node [style=none] (8) at (1.5, 0) {$=$};
		\node [style=rn] (9) at (0, 0.25) {};
		\node [style=gn] (10) at (0, -0.25) {};
	\end{pgfonlayer}
	\begin{pgfonlayer}{edgelayer}
		\draw [style=none] (5) to (0);
		\draw[bend right=23]  [style=none] (0) to (2.center);
		\draw[bend left=23]  [style=none] (0) to (4.center);
		\draw [style=none] (3) to (1.center);
		\draw [style=none] (6) to (7.center);
		\draw (9) to (10);
	\end{pgfonlayer}
\end{tikzpicture}&(B1) && \begin{tikzpicture}
	\begin{pgfonlayer}{nodelayer}
		\node [style=none] (0) at (1.5, 0.75) {};
		\node [style=rn] (1) at (-1.75, -0.2500001) {};
		\node [style=none] (2) at (1, -0.7499998) {};
		\node [style=none] (3) at (-1, 0.9999999) {};
		\node [style=none] (4) at (1, 0.75) {};
		\node [style=none] (5) at (-1.75, -0.7499998) {};
		\node [style=none] (6) at (-1.75, 0.9999999) {};
		\node [style=gn] (7) at (-1, 0.5) {};
		\node [style=none] (8) at (0, -0) {$=$};
		\node [style=rn] (9) at (-1, -0.2500001) {};
		\node [style=gn] (10) at (1.25, -0.2500001) {};
		\node [style=gn] (11) at (-1.75, 0.5) {};
		\node [style=none] (12) at (-1, -0.7499998) {};
		\node [style=none] (13) at (1.5, -0.7499998) {};
		\node [style=rn] (14) at (1.25, 0.2500001) {};
		\node [style=rn] (15) at (-2.25, 0.2500001) {};
		\node [style=gn] (16) at (-2.25, -0.2500001) {};
		\node [style=none] (17) at (0, -0.9999999) {};
		\node [style=none] (18) at (0, 1.25) {};
	\end{pgfonlayer}
	\begin{pgfonlayer}{edgelayer}
		\draw [style=none] (12.center) to (9);
		\draw [style=none] (5.center) to (1);
		\draw [style=none] (7) to (3.center);
		\draw [style=none, bend right=23, looseness=1.00] (9) to (7);
		\draw [style=none] (11) to (6.center);
		\draw [style=none, bend left=23, looseness=1.00] (1) to (11);
		\draw [style=none, bend right=23, looseness=1.00] (13.center) to (10);
		\draw [style=none] (10) to (14);
		\draw [style=none, bend left=23, looseness=1.00] (14) to (4.center);
		\draw [style=none, bend right=23, looseness=1.00] (14) to (0.center);
		\draw [bend right=23, looseness=1.00] (10) to (2.center);
		\draw (11) to (9);
		\draw (7) to (1);
		\draw (15) to (16);
	\end{pgfonlayer}
\end{tikzpicture}&(B2)\\
   &&&& \\
   \begin{tikzpicture}
	\begin{pgfonlayer}{nodelayer}
		\node [style=gn] (0) at (0.5, -0.55) {$~\frac{\pi}{2}~$};
		\node [style=rn] (1) at (0.5, -0) {};
		\node [style=gn] (2) at (0.5, 0.55) {$~\frac{\pi}{2}~$};
		\node [style=none] (3) at (0.5, 0.9999999) {};
		\node [style=gn] (4) at (1.25, 0.5) {$\frac{-\pi}{2}$};
		\node [style=none] (5) at (0.5, -0.9999999) {};
		\node [style=none] (6) at (-0.9999999, 0.9999999) {};
		\node [style=none] (7) at (-0.9999999, -0.9999999) {};
		\node [style=none] (8) at (-0.2500001, -0) {$=$};
		\node [style={{H box}}] (9) at (-0.9999999, -0) {};
	\end{pgfonlayer}
	\begin{pgfonlayer}{edgelayer}
		\draw (3.center) to (2);
		\draw (2) to (1);
		\draw (1) to (0);
		\draw (0) to (5.center);
		\draw (1) to (4);
		\draw (6.center) to (7.center);
	\end{pgfonlayer}
\end{tikzpicture}&(EU) &&  \begin{tikzpicture}
	\begin{pgfonlayer}{nodelayer}
		\node [style=none] (0) at (0.5000002, 0.7499999) {};
		\node [style=none] (1) at (-0.7500001, -1) {};
		\node [style={H box}] (2) at (-0.7500001, 0.5) {};
		\node [style=none] (3) at (-1.25, -0.7499999) {\raisebox{2mm}{...}};
		\node [style=none] (4) at (0.9999999, -0.7499999) {\raisebox{2mm}{...}};
		\node [style=none] (5) at (0.9999999, 0.7499999) {\raisebox{-2mm}{...}};
		\node [style=rn] (6) at (-1.25, -0) {\footnotesize$~\alpha~$};
		\node [style={H box}] (7) at (-1.75, -0.5000001) {};
		\node [style=none] (8) at (-0.7500001, 1) {};
		\node [style={H box}] (9) at (-0.7500001, -0.5000001) {};
		\node [style=none] (10) at (-1.25, 0.7499999) {\raisebox{-2mm}{...}};
		\node [style=none] (11) at (0, -0) {$=$};
		\node [style=gn] (12) at (0.9999999, -0) {\footnotesize$~\alpha~$};
		\node [style=none] (13) at (1.5, -0.7499999) {};
		\node [style=none] (14) at (1.5, 0.7499999) {};
		\node [style=none] (15) at (-1.75, 1) {};
		\node [style={H box}] (16) at (-1.75, 0.5) {};
		\node [style=none] (17) at (-1.75, -1) {};
		\node [style=none] (18) at (0.5000002, -0.7499999) {};
	\end{pgfonlayer}
	\begin{pgfonlayer}{edgelayer}
		\draw [bend right, looseness=1.00] (6) to (7);
		\draw [bend left, looseness=1.00] (6) to (9);
		\draw (9) to (1.center);
		\draw (7) to (17.center);
		\draw [bend left, looseness=1.00] (6) to (16);
		\draw [bend right, looseness=1.00] (6) to (2);
		\draw (2) to (8.center);
		\draw (16) to (15.center);
		\draw [bend left=23, looseness=1.00] (12) to (0.center);
		\draw [bend right=23, looseness=1.00] (12) to (14.center);
		\draw [bend right=23, looseness=1.00] (12) to (18.center);
		\draw [bend left=23, looseness=1.00] (12) to (13.center);
	\end{pgfonlayer}
\end{tikzpicture}&(H)\\
   &&&& \\
   \begin{tikzpicture}
	\begin{pgfonlayer}{nodelayer}
		\node [style=none] (0) at (-0.25, -0) {$=$};
		\node [style=none] (1) at (-1, 1) {};
		\node [style=none] (2) at (-1, -0.75) {};
		\node [style=none] (3) at (1.25, 1) {};
		\node [style=gn] (4) at (-1, -0.25) {$~\pi~$};
		\node [style=none] (5) at (1.25, -0.7499999) {};
		\node [style=rn] (6) at (-1, 0.5) {$~\alpha~$};
		\node [style=rn] (7) at (1.25, -0.25) {$-\alpha$};
		\node [style=gn] (8) at (1.25, 0.5) {$~\pi~$};
		\node [style=rn] (9) at (-1.75, 0.25) {};
		\node [style=gn] (10) at (-1.75, -0.25) {};
		\node [style=rn] (11) at (0.5000002, 0.5) {$~\alpha~$};
		\node [style=gn] (12) at (0.5000002, -0.25) {$~\pi~$};
	\end{pgfonlayer}
	\begin{pgfonlayer}{edgelayer}
		\draw (3.center) to (8);
		\draw (8) to (7);
		\draw (7) to (5.center);
		\draw (4) to (2.center);
		\draw (1.center) to (6);
		\draw (6) to (4);
		\draw (9) to (10);
		\draw (12) to (11);
	\end{pgfonlayer}
\end{tikzpicture}&(K2) && \begin{tikzpicture}
	\begin{pgfonlayer}{nodelayer}
		\node [style=gn] (0) at (-1.25, 0.7499999) {$~~\alpha~~$};
		\node [style=gn] (1) at (-0.2499998, 0.7499999) {$\alpha{+}\pi$};
		\node [style=gn] (2) at (1.25, 0.7499999) {$2\alpha{+}\pi$};
		\node [style=none] (3) at (1.25, -0.7500001) {};
		\node [style=rn] (4) at (-0.7500001, -0.2500001) {};
		\node [style=none] (5) at (-0.7499998, -0.7499998) {};
		\node [style=none] (6) at (0.2499998, -0) {$=$};
		\node [style=rn] (7) at (1.25, -0.2500001) {};
	\end{pgfonlayer}
	\begin{pgfonlayer}{edgelayer}
		\draw (0) to (4);
		\draw (1) to (4);
		\draw (4) to (5.center);
		\draw (7) to (3.center);
		\draw [bend right=45, looseness=1.00] (2) to (7);
		\draw [bend right=45, looseness=1.00] (7) to (2);
	\end{pgfonlayer}
\end{tikzpicture}&(SUP) \\
   &&&& \\
   \hline
  \end{tabular}
 \caption[]{Set of rules for the ZX-calculus \cite{supplementarity} with scalars. All of these rules also hold when flipped upside-down, or with the colours red and green swapped. The right-hand side of (IV) is an empty diagram. ($\cdots$) denote zero or more wires, while (\protect\rotatebox{45}{\raisebox{-0.4em}{$\cdots$}}) denote one or more wires.}
 \label{fig:ZX_rules}
\end{figure}
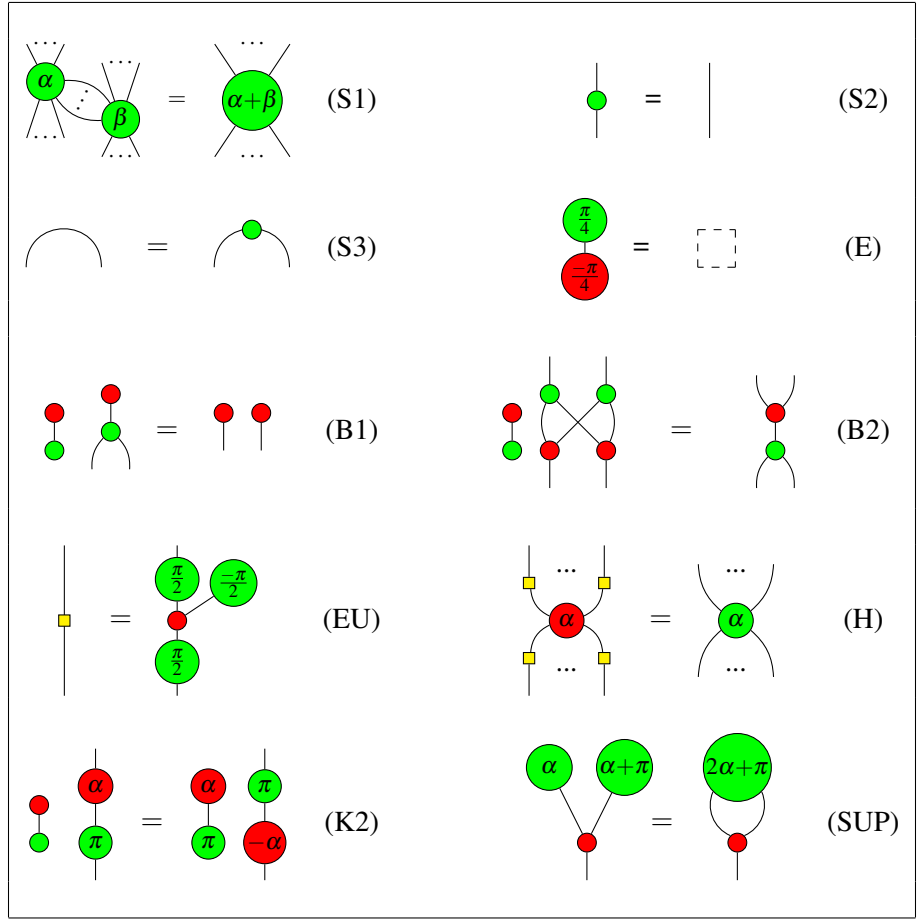

The standard interpretation of the ZX-diagrams associates with any diagram $D:n\to m$ a linear map $\interp{D}:\mathbb{C}^{2^n}\to\mathbb{C}^{2^m}$ inductively defined as follows:
$$ \interp{D_1\otimes D_2}:=\interp{D_1}\otimes\interp{D_2} \qquad 
\interp{D_2\circ D_1}:=\interp{D_2}\circ\interp{D_1} \qquad
\interp{
\text{ and }M^{\otimes k}=M\otimes M^{\otimes k-1}\text{ for any }k\in \mathbb{N}^*\right)$

The transformation rules of the ZX-calculus are expressed in the figure \ref{fig:ZX_rules}, \cite{gen-supp}.  Notice that the rule \zxe requires the angles $\pm\pi/4$. When a restriction of the language that does not include the angles $\pm\pi/4$ is considered, the rule \zxe is to be replaced by \zxzo and \zxiv:

\[ \hypertarget{r:zx2}{}%
\quad\textnormal{(ZO)}\]
\subsection{\reca}

A \reda $D:k\to l$ is an open diagram with $k$ inputs and $l$ outputs and is generated by:\\
\begin{center}
\bgroup
\def\arraystretch{2.5}
\begin{tabular}{|cc|cc|}
\hline
$R_Z^{(n,m)}:n\to m$ & \begin{tikzpicture}
	\begin{pgfonlayer}{nodelayer}
		\node [style=gn] (0) at (0, -0) {};
		\node [style=none] (1) at (-0.5000001, 0.7499999) {};
		\node [style=none] (2) at (0, 0.4999999) {$\cdots$};
		\node [style=none] (3) at (0.5000001, 0.7499999) {};
		\node [style=none] (4) at (-0.5000001, -0.7499999) {};
		\node [style=none] (5) at (0, -0.4999999) {$\cdots$};
		\node [style=none] (6) at (0.5000001, -0.7499999) {};
		\node [style=none] (7) at (0, 0.75) {$n$};
		\node [style=none] (8) at (0, -0.7500001) {$m$};
		\node [style=none] (9) at (0, 0.9999999) {};
		\node [style=none] (10) at (0, -0.9999999) {};
	\end{pgfonlayer}
	\begin{pgfonlayer}{edgelayer}
		\draw [style=none, bend right, looseness=1.00] (1.center) to (0);
		\draw [style=none, bend right, looseness=1.00] (0) to (3.center);
		\draw [style=none, bend right, looseness=1.00] (0) to (4.center);
		\draw [style=none, bend left, looseness=1.00] (0) to (6.center);
	\end{pgfonlayer}
\end{tikzpicture} & $R_X^{(n,m)}:n\to m$ & \begin{tikzpicture}
	\begin{pgfonlayer}{nodelayer}
		\node [style=rn] (0) at (0, -0) {};
		\node [style=none] (1) at (-0.5000001, 0.7499999) {};
		\node [style=none] (2) at (0, 0.4999999) {$\cdots$};
		\node [style=none] (3) at (0.5000001, 0.7499999) {};
		\node [style=none] (4) at (-0.5000001, -0.7499999) {};
		\node [style=none] (5) at (0, -0.4999999) {$\cdots$};
		\node [style=none] (6) at (0.5000001, -0.7499999) {};
		\node [style=none] (7) at (0, 0.75) {$n$};
		\node [style=none] (8) at (0, -0.7500001) {$m$};
	\end{pgfonlayer}
	\begin{pgfonlayer}{edgelayer}
		\draw [style=none, bend right, looseness=1.00] (1.center) to (0);
		\draw [style=none, bend right, looseness=1.00] (0) to (3.center);
		\draw [style=none, bend right, looseness=1.00] (0) to (4.center);
		\draw [style=none, bend left, looseness=1.00] (0) to (6.center);
	\end{pgfonlayer}
\end{tikzpicture}\\
\hline
\end{tabular}
\egroup\\
where $n,m\in \mathbb{N}$ and $\alpha \in \mathbb{R}$
\end{center}

\begin{itemize}
\item Spatial Composition: for any $D_1:a\to b$ and $D_2:c\to d$, $D_1\otimes D_2:a+c\to b+d$ consists in placing $D_1$ and $D_2$ side by side, $D_2$ on the right of $D_1$.
\item Sequential Composition: for any $D_1:a\to b$ and $D_2:b\to c$, $D_2\circ D_1:a\to c$ consists in placing $D_1$ on the top of $D_2$, connecting the outputs of $D_1$ to the inputs of $D_2$.
\end{itemize}

\begin{figure}[!htb]
 \centering
 \hypertarget{r:rules}{}
  \begin{tabular}{|ccccc|}
   \hline
   \begin{tikzpicture}[font={\footnotesize}]
	\begin{pgfonlayer}{nodelayer}
		\node [style=none] (0) at (-1.25, -0) {\rotatebox[origin=c]{63.43}{...}};
		\node [style=none] (1) at (0.25, -0) {$=$};
		\node [style=gn] (2) at (1.5, 0) {};
		\node [style=gn] (3) at (-0.75, -0.25) {};
		\node [style=none] (4) at (-1.75, -0.5) {\raisebox{2mm}{...}};
		\node [style=none] (5) at (2, -0.75) {};
		\node [style=none] (6) at (-1, -0.75) {};
		\node [style=none] (7) at (1.5, -0.75) {\raisebox{2mm}{...}};
		\node [style=none] (8) at (-0.5, -0.75) {};
		\node [style=none] (9) at (1, -0.75) {};
		\node [style=none] (10) at (-2, -0.5) {};
		\node [style=none] (11) at (-1.5, -0.5) {};
		\node [style=none] (12) at (-0.75, -0.75) {\raisebox{2mm}{...}};
		\node [style=none] (13) at (1.5, 0.75) {\raisebox{-2mm}{...}};
		\node [style=none] (14) at (1, 0.75) {};
		\node [style=none] (15) at (-2, 0.75) {};
		\node [style=none] (16) at (-0.5, 0.5) {};
		\node [style=none] (17) at (-1.5, 0.75) {};
		\node [style=none] (18) at (2, 0.75) {};
		\node [style=gn] (19) at (-1.75, 0.25) {};
		\node [style=none] (20) at (-0.75, 0.5) {\raisebox{-2mm}{...}};
		\node [style=none] (21) at (-1, 0.5) {};
		\node [style=none] (22) at (-1.75, 0.75) {\raisebox{-2mm}{...}};
		\node [style=none] (23) at (1, 0.9999999) {};
		\node [style=none] (24) at (1, -0.9999999) {};
	\end{pgfonlayer}
	\begin{pgfonlayer}{edgelayer}
		\draw (3) to (16.center);
		\draw (3) to (6.center);
		\draw (3) to (8.center);
		\draw (19) to (10.center);
		\draw (19) to (11.center);
		\draw [bend right, looseness=1.00] (19) to (3);
		\draw [bend left, looseness=1.00] (19) to (3);
		\draw (14.center) to (2);
		\draw (2) to (9.center);
		\draw (5.center) to (2);
		\draw (2) to (18.center);
		\draw (19) to (15.center);
		\draw (19) to (17.center);
		\draw (3) to (21.center);
	\end{pgfonlayer}
\end{tikzpicture}&(S1) &$\qquad$ & \begin{tikzpicture}
	\begin{pgfonlayer}{nodelayer}
		\node [style=gn] (0) at (-0.7499998, -0) {};
		\node [style=none] (1) at (0, -0) {=};
		\node [style=none] (2) at (-0.7499998, 0.5) {};
		\node [style=none] (3) at (-0.7499998, -0.5000001) {};
		\node [style=none] (4) at (0.7499998, 0.5) {};
		\node [style=none] (5) at (0.7499998, -0.5000001) {};
		\node [style=none] (6) at (0, -0.7499998) {};
		\node [style=none] (7) at (0, 0.75) {};
	\end{pgfonlayer}
	\begin{pgfonlayer}{edgelayer}
		\draw (2.center) to (0);
		\draw (0) to (3.center);
		\draw (4.center) to (5.center);
	\end{pgfonlayer}
\end{tikzpicture}&(S2)\\
\hline
\end{tabular}
 \caption[]{Rules for the \textbf{\reca} with scalars. All of these rules also hold when flipped upside-down, or with the colours red and green swapped and the real-boxes flipped. The right-hand side of (IV) is an empty diagram. ($\cdots$) denote zero or more wires, while (\protect\rotatebox{45}{\raisebox{-0.4em}{$\cdots$}}) denote one or more wires.}
 \label{fig:Y_rules}
\end{figure}

The standard interpretation of the \redas associates any diagram $D:n\to m$ with a linear map $\interp{D}:\mathbb{R}^{2^n}\to\mathbb{R}^{2^m}$ inductively defined as follows:
$$%
\text{ and }M^{\otimes k}=M\otimes M^{\otimes k-1}\text{ for any }k\in \mathbb{N}^*\right)$.

We define a set of basic transformations of \redas that preserve the matrices they represent. These axioms are expressed in figure \ref{fig:Y_rules}, where the upside-down box is defined as:
\[ %
 \]

\subsection{In both calculi}

\textbf{Only Topology matters} is a paradigm -- provable in both the ZX-Calculus and the \reca -- stating that one can bend or stretch the wires at will.

\begin{example}
\[%
\]
\end{example}
Therefore, two vertices connected by an horizontal wire have meaning.

\begin{theorem}
All the equalities in Figures \ref{fig:ZX_rules} and \ref{fig:Y_rules} are sound, i.e.~for $L\in\{ZX;\rec\}$ $$(L\vdash D_1=D_2)\implies(\interp{D_1}=\interp{D_2})$$
\end{theorem}

When we can show that a diagram $D_1$ is equal to another one, $D_2$, using a succession of equalities of the set of rules $L\in\{ZX;\rec\}$, we write $L\vdash D_1 = D_2$. Given that the rules are sound, this implies that $\interp{D_1} = \interp{D_2}$. The rules can obviously be applied to any subdiagram, meaning, for any diagram $D$:
\[ (L\vdash D_1=D_2)\implies \left\lbrace\begin{array}{ccc}
(L\vdash D_1\circ D = D_2\circ D) & \land & (L\vdash D\circ D_1 = D\circ D_2)\\
(L\vdash D_1\otimes D = D_2\otimes D) & \land & (L\vdash D\otimes D_1 = D\otimes D_2)
\end{array}\right. \]

\subsection{Discussion on the ``real boxes''}

\textbf{Directedness:}
The real boxes represent real rotations. Unlike complex rotations -- such as the ones induced by the green and red dots --, their corresponding matrices cannot be symmetrical. Indeed, a real symmetrical matrix is diagonalisable, and rotation matrices are orthogonal. However the only real diagonal and orthogonal matrices have diagonal coefficients in $\{-1;1\}$, hence, representing a rotation of angle $\alpha$ with a real symmetrical matrix would be impossible.
~\\
\noindent\textbf{$4\pi$-periodicity:} Textbook definitions of quantum mechanics rotation operators are often $4\pi$-periodical -- see for instance Nielsen and Chuang's \cite{nielsen_chuang_2010}: given an operator $A$ s.t. $A^2=I$ one can define the rotation $R_A(\alpha)= \cos(\frac \alpha 2)I-i\sin(\frac\alpha2) A$ which satisfies $R_A(2\pi)=-1$. The interpretation of this non $2\pi$-periodicity is known as the orientation entanglement \cite{feynman}.
Real rotations of the Y-calculus correspond to the case $A=Y = \left(\begin{array}{cc}0&-i\\i&0\end{array}\right)$.
In the ZX-Calculus, rotations have been made $2\pi$-periodical \cite{interacting} by considering the operator $e^{i\frac \alpha 2}R_A(\alpha)$ instead of $R_A(\alpha)$. However, one cannot do the same with real rotations.

\begin{proposition}[$4\pi$-periodicity]
\label{prop:4_pi-periodical}
The real boxes are $4\pi$-periodical:
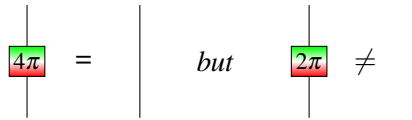
\[\begin{tikzpicture}
	\begin{pgfonlayer}{nodelayer}
		\node [style=bx] (0) at (-2.5, -0) {$4\pi$};
		\node [style=none] (1) at (-2.5, 0.7499999) {};
		\node [style=none] (2) at (-2.5, -0.7500001) {};
		\node [style=none] (3) at (-1, 0.7499999) {};
		\node [style=none] (4) at (-1, -0.7500001) {};
		\node [style=none] (5) at (-1.75, -0) {=};
		\node [style=none] (6) at (2.75, -0.7500001) {};
		\node [style=bx] (7) at (1.25, -0) {$2\pi$};
		\node [style=none] (8) at (1.25, -0.7500001) {};
		\node [style=none] (9) at (1.25, 0.7499999) {};
		\node [style=none] (10) at (2, -0) {$\neq$};
		\node [style=none] (11) at (2.75, 0.7499999) {};
		\node [style=none] (12) at (0, -0) {but};
	\end{pgfonlayer}
	\begin{pgfonlayer}{edgelayer}
		\draw (1.center) to (2.center);
		\draw (3.center) to (4.center);
		\draw (9.center) to (8.center);
		\draw (11.center) to (6.center);
	\end{pgfonlayer}
\end{tikzpicture}\]
\end{proposition}

\begin{proof}
In appendix at page \pageref{prf:4_pi-periodical}.
\end{proof}

\section{Minimality}
\label{sec:minimality}

In this section, we prove the necessity of some rules of the \reca i.e.~we show that some of its axioms are not deducible from the others. A rule $(R)$ is necessary when $Y\setminus\{(R)\} \nvdash (R)$.

\begin{proposition}~~\\
\label{prop:rs3-necessary}
\vspace{-2em}
\begin{multicols}{2}
\begin{tikzpicture}
	\begin{pgfonlayer}{nodelayer}
		\node [style=gn] (0) at (-2.25, -0.5000001) {};
		\node [style=gn] (1) at (-2.25, 0.9999999) {};
		\node [style=gn] (2) at (-1.5, 0.2500001) {};
		\node [style=bx] (3) at (-2.25, -0.9999999) {$\alpha$};
		\node [style=none] (4) at (-2.25, 1.5) {};
		\node [style=none] (5) at (-0.7499998, 1.5) {};
		\node [style=none] (6) at (-2.25, -1.5) {};
		\node [style=none] (7) at (0, -0) {=};
		\node [style=udbx] (8) at (2.5, 0.9999999) {$\alpha$};
		\node [style=bx] (9) at (-1.5, -0.2500001) {$\pi/2$};
		\node [style=udbx] (10) at (-2.25, 0.2500001) {$\pi/2$};
		\node [style=bx] (11) at (-1.5, 0.75) {$\pi/2$};
		\node [style=none] (12) at (1, 1.5) {};
		\node [style=bx] (13) at (1.75, -0.2500001) {$\pi/2$};
		\node [style=none] (14) at (2.5, 1.5) {};
		\node [style=gn] (15) at (1, 0.9999999) {};
		\node [style=gn] (16) at (1.75, 0.2500001) {};
		\node [style=none] (17) at (1, -1.5) {};
		\node [style=udbx] (18) at (1, 0.2500001) {$\pi/2$};
		\node [style=gn] (19) at (1, -0.5000001) {};
		\node [style=bx] (20) at (1.75, 0.75) {$\pi/2$};
		\node [style=none] (21) at (1, -1.75) {};
		\node [style=none] (22) at (1, 1.75) {};
	\end{pgfonlayer}
	\begin{pgfonlayer}{edgelayer}
		\draw (4.center) to (1);
		\draw (0) to (3);
		\draw (3) to (6.center);
		\draw [in=-1, out=-90, looseness=1.25] (5.center) to (2);
		\draw (1) to (10);
		\draw (10) to (0);
		\draw [in=-102, out=-45, looseness=1.25] (0) to (9);
		\draw (9) to (2);
		\draw [style=none] (11) to (2);
		\draw [style=none, in=90, out=45, looseness=1.25] (1) to (11);
		\draw (12.center) to (15);
		\draw (15) to (18);
		\draw (18) to (19);
		\draw [in=-102, out=-45, looseness=1.25] (19) to (13);
		\draw (13) to (16);
		\draw [style=none] (20) to (16);
		\draw [style=none, in=90, out=45, looseness=1.25] (15) to (20);
		\draw [style=none] (14.center) to (8);
		\draw [style=none, in=0, out=-90, looseness=1.25] (8) to (16);
		\draw [style=none] (19) to (17.center);
	\end{pgfonlayer}
\end{tikzpicture}$\quad$\textnormal{(RS3)}\\
\vfill \null 
\noindent
cannot be derived from the other rules in any $\frac{\pi}{2n}$-fragment $(n\in\mathbb{N}^*)$.
\vfill \null 
\end{multicols}
\end{proposition}

\begin{proof}
In appendix at page \pageref{prf:rs3-necessary}.
\end{proof}

\begin{proposition}~~\\
\label{prop:rh-necessary}
\vspace{-2em}
\begin{multicols}{2}
\begin{tikzpicture}
	\begin{pgfonlayer}{nodelayer}
		\node [style=gn] (0) at (-1.25, -0) {};
		\node [style=none] (1) at (0, -0) {=};
		\node [style=rn] (2) at (1.249999, -0) {};
		\node [style=bx] (3) at (-2, 0.75) {$\pi/2$};
		\node [style=bx] (4) at (-0.5000001, 0.75) {$\pi/2$};
		\node [style=udbx] (5) at (-2, -0.75) {$\pi/2$};
		\node [style=udbx] (6) at (-0.5000001, -0.75) {$\pi/2$};
		\node [style=none] (7) at (-2, 1.25) {};
		\node [style=none] (8) at (-0.5000001, 1.25) {};
		\node [style=none] (9) at (-2, -1.25) {};
		\node [style=none] (10) at (-0.5000001, -1.25) {};
		\node [style=none] (11) at (0.5000001, -1.25) {};
		\node [style=none] (12) at (0.5000001, 1.25) {};
		\node [style=none] (13) at (2, 1.25) {};
		\node [style=none] (14) at (2, -1.25) {};
		\node [style=none] (15) at (1.249999, -0.75) {$\cdots$};
		\node [style=none] (16) at (1.249999, 0.75) {$\cdots$};
		\node [style=none] (17) at (-1.25, 0.75) {$\cdots$};
		\node [style=none] (18) at (-1.25, -0.75) {$\cdots$};
	\end{pgfonlayer}
	\begin{pgfonlayer}{edgelayer}
		\draw [bend right=45, looseness=1.00] (3) to (0);
		\draw [bend right=45, looseness=1.25] (0) to (5);
		\draw (5) to (9.center);
		\draw (6) to (10.center);
		\draw [bend right=45, looseness=1.00] (0) to (4);
		\draw (3) to (7.center);
		\draw (4) to (8.center);
		\draw [bend right, looseness=0.75] (12.center) to (2);
		\draw [bend right, looseness=0.75] (2) to (13.center);
		\draw [bend left=45, looseness=1.25] (0) to (6);
		\draw [bend right, looseness=1.00] (2) to (11.center);
		\draw [bend left, looseness=1.00] (2) to (14.center);
	\end{pgfonlayer}
\end{tikzpicture}$\quad$\textnormal{(RH)}\\
\vfill \null \noindent
cannot be derived from the other rules.
\vfill \null
\end{multicols}
\end{proposition}

\begin{proof}
In appendix at page \pageref{prf:rh-necessary}.
\end{proof}

\section{Completeness of the \frag{2}}
\label{sec:pi_2-completeness}
The \frag{2} of the ZX-Calculus has been proven to be complete \cite{pi_2-complete}. We can prove the same result with the \reca, though it only makes use of the completeness of the $\pi$-fragment of the ZX-Calculus (ZX$_r$) \cite{pivoting}, defined as:

\begin{definition}
\label{def:zxr}
The ZX$_r$-diagrams are generated in the same way as ZX-diagrams, but with angles in $\{0,\pi\}$. Its set of rules is defined as:
\[ZX_r = \{\zxzo,\zxiv,\zxrhl\}\cup ZX\setminus\{\zxe,\zxsup,\zxeu,\zxkt\}\]
with $\qquad\hypertarget{r:zxr-rules}{}\begin{tikzpicture}
	\begin{pgfonlayer}{nodelayer}
		\node [style=gn] (0) at (-0.7499998, -0) {$~\pi~$};
		\node [style=none] (1) at (-0.7499998, 0.75) {};
		\node [style=none] (2) at (-0.7499998, -0.75) {};
		\node [style=none] (3) at (0, -0) {=};
		\node [style=gn] (4) at (0.7499998, -0) {};
		\node [style=none] (5) at (0.7499998, -0.75) {};
		\node [style=none] (6) at (0.7499998, 0.75) {};
		\node [style={{H box}}] (7) at (1.25, -0) {};
		\node [style=gn] (8) at (1.75, 0.2500001) {};
		\node [style=rn] (9) at (1.75, -0.2499999) {};
	\end{pgfonlayer}
	\begin{pgfonlayer}{edgelayer}
		\draw (1.center) to (2.center);
		\draw (6.center) to (5.center);
		\draw [bend left=90, looseness=2.00] (4) to (7);
		\draw [bend left=90, looseness=1.75] (7) to (4);
		\draw (8) to (9);
	\end{pgfonlayer}
\end{tikzpicture}\quad\textnormal{(HL)}$
\end{definition}

\begin{theorem}
\label{th:pi_2-complete}
The \frag{2} of the \reca (\repiot) is complete.
\end{theorem}

\begin{proof}
\phantomsection\label{prf:pi_2-complete}
The idea of the proof is to show that \repiot and the real stabiliser ZX-Calculus (ZX$_r$) \cite{pivoting} deal with the same matrices and have the same expressivity.

To do so, we define the interpretations:
\[ \fit{$\interp{.}^{\repiot \to ZX_r}: \left\lbrace \begin{array}{l}
D_1\circ D_2 \mapsto \interp{D_1}^{\repiot \to ZX_r} \circ \interp{D_2}^{\repiot \to ZX_r}\\
D_1\otimes D_2 \mapsto \interp{D_1}^{\repiot \to ZX_r} \otimes \interp{D_2}^{\repiot \to ZX_r}\\\\
\begin{tikzpicture}
	\begin{pgfonlayer}{nodelayer}
		\node [style=bx] (0) at (-1, 0.9999999) {$k\frac{\pi}{2}$};
		\node [style=none] (1) at (-1, 0.5000001) {};
		\node [style=none] (2) at (-1, 1.5) {};
		\node [style=none] (3) at (0, 0.9999999) {$\mapsto$};
		\node [style=none] (4) at (1.25, 1.5) {};
		\node [style=none] (5) at (1.25, 0.5000001) {};
		\node [style=none, anchor=west] (6) at (1.25, 1.5) {\scriptsize $~\circ k$};
		\node [style=gn] (7) at (1, 1.25) {$\pi$};
		\node [style={{H box}}] (8) at (1, 0.7499999) {};
		\node [style=none] (9) at (1, 1.75) {};
		\node [style=none] (10) at (1, 0.25) {};
		\node [style=none] (11) at (0.7500001, 1.5) {};
		\node [style=none] (12) at (0.7500001, 0.5000001) {};
		\node [style={{H box}}] (13) at (1, -0.5000001) {};
		\node [style=none] (14) at (1.25, -1.25) {};
		\node [style=none] (15) at (1, -1.5) {};
		\node [style=none] (16) at (0.7500001, -1.25) {};
		\node [style=none] (17) at (1.25, -0.25) {};
		\node [style=none] (18) at (0.7500001, -0.25) {};
		\node [style=none] (19) at (0, -0.7499999) {$\mapsto$};
		\node [style=none] (20) at (1, -0) {};
		\node [style=none, anchor=west] (21) at (1.25, -0.25) {\scriptsize $~\circ k$};
		\node [style=udbx] (22) at (-1, -0.7499999) {$k\frac{\pi}{2}$};
		\node [style=gn] (23) at (1, -1) {$\pi$};
		\node [style=none] (24) at (-1, -1.25) {};
		\node [style=none] (25) at (-1, -0.2499998) {};
	\end{pgfonlayer}
	\begin{pgfonlayer}{edgelayer}
		\draw (2.center) to (1.center);
		\draw [bend right=45, looseness=0.50] (5.center) to (4.center);
		\draw (9.center) to (10.center);
		\draw [bend right=45, looseness=0.50] (11.center) to (12.center);
		\draw (25.center) to (24.center);
		\draw [bend right=45, looseness=0.50] (14.center) to (17.center);
		\draw (20.center) to (15.center);
		\draw [bend right=45, looseness=0.50] (18.center) to (16.center);
	\end{pgfonlayer}
\end{tikzpicture}\\\\
Id\quad \text{otherwise}
\end{array}\right. \hspace{-2em}
\interp{.}^{ZX_r \to \repiot}: \left\lbrace \begin{array}{l}
D_1\circ D_2 \mapsto \interp{D_1}^{ZX_r \to \repiot} \circ \interp{D_2}^{ZX_r \to \repiot}\\
D_1\otimes D_2 \mapsto \interp{D_1}^{ZX_r \to \repiot} \otimes \interp{D_2}^{ZX_r \to \repiot}\\\\
\begin{tikzpicture}
	\begin{pgfonlayer}{nodelayer}
		\node [style=bx] (0) at (1, -1.5) {$\pi/2$};
		\node [style=none] (1) at (1, -2) {};
		\node [style=none] (2) at (1, -0.4999999) {};
		\node [style=none] (3) at (0, 1.5) {$\mapsto$};
		\node [style=gn] (4) at (-1, 1.5) {$\pi$};
		\node [style=none] (5) at (-1.25, 2) {};
		\node [style=none] (6) at (-0.7500001, 0.9999999) {};
		\node [style={{H box}}] (7) at (-1, -1.25) {};
		\node [style=none] (8) at (-1, -2) {};
		\node [style=none] (9) at (0, -1.25) {$\mapsto$};
		\node [style=none] (10) at (-1, -0.4999999) {};
		\node [style=gn] (11) at (1, 1.5) {};
		\node [style=gn] (12) at (2, 1.5) {};
		\node [style=rbx] (13) at (1.5, 1.5) {$\pi$};
		\node [style=rbx] (14) at (1.5, -0.9999999) {$\pi$};
		\node [style=gn] (15) at (2, -0.9999999) {};
		\node [style=gn] (16) at (1, -0.9999999) {};
		\node [style=none] (17) at (-0.7500001, 2) {};
		\node [style=none] (18) at (-1.25, 0.9999999) {};
		\node [style=none] (19) at (1.25, 2) {};
		\node [style=none] (20) at (1.25, 0.9999999) {};
		\node [style=none] (21) at (0.7500001, 2) {};
		\node [style=none] (22) at (0.7500001, 0.9999999) {};
		\node [style=none] (23) at (1.25, -0.2499998) {};
		\node [style=rn] (24) at (-1, 0.25) {$\pi$};
		\node [style=none] (25) at (0.7500001, 0.7500001) {};
		\node [style=lbx] (26) at (1.5, 0.25) {$\pi$};
		\node [style=none] (27) at (0, 0.25) {$\mapsto$};
		\node [style=none] (28) at (1.25, 0.7500001) {};
		\node [style=none] (29) at (-1.25, -0.2499998) {};
		\node [style=none] (30) at (-0.7500001, -0.2499998) {};
		\node [style=none] (31) at (-1.25, 0.7500001) {};
		\node [style=none] (32) at (0.7500001, -0.2499998) {};
		\node [style=rn] (33) at (1, 0.25) {};
		\node [style=none] (34) at (-0.7500001, 0.7500001) {};
		\node [style=rn] (35) at (2, 0.25) {};
		\node [style=none] (36) at (-1, 2) {$~\cdots~$};
		\node [style=none] (37) at (-1, 0.9999999) {$~\cdots~$};
		\node [style=none] (38) at (1, 2) {$~\cdots~$};
		\node [style=none] (39) at (-1, 0.7500001) {$~\cdots~$};
		\node [style=none] (40) at (1, 0.9999999) {$~\cdots~$};
		\node [style=none] (41) at (1, 0.7500001) {$~\cdots~$};
		\node [style=none] (42) at (-1, -0.2499998) {$~\cdots~$};
		\node [style=none] (43) at (1, -0.2499998) {$~\cdots~$};
	\end{pgfonlayer}
	\begin{pgfonlayer}{edgelayer}
		\draw (2.center) to (1.center);
		\draw [in=90, out=-90, looseness=0.75] (5.center) to (6.center);
		\draw (10.center) to (8.center);
		\draw (12) to (11);
		\draw (15) to (16);
		\draw [in=90, out=-90, looseness=0.75] (17.center) to (18.center);
		\draw [in=90, out=-90, looseness=0.75] (21.center) to (20.center);
		\draw [in=90, out=-90, looseness=0.75] (19.center) to (22.center);
		\draw [in=90, out=-90, looseness=0.75] (31.center) to (30.center);
		\draw (35) to (33);
		\draw [in=90, out=-90, looseness=0.75] (34.center) to (29.center);
		\draw [in=90, out=-90, looseness=0.75] (25.center) to (23.center);
		\draw [in=90, out=-90, looseness=0.75] (28.center) to (32.center);
	\end{pgfonlayer}
\end{tikzpicture}\\\\
Id\quad \text{otherwise}
\end{array}\right.$}\]
for $k\geq 0$ with $D^{\circ 0} = \mathbb{I}$ and $D^{\circ l} = D^{\circ l-1}\circ D$ for $l\geq 2$.\\
It is important to notice that the rule \rsup is not an axiom of the language \repiot. Indeed, \rsupc{2} can be derived from the other rules whenever $\alpha$ is a multiple of $\frac{\pi}{2}$.\\

The two interpretations both preserve the equalities of the sets of rules of respectively \repiot and ZX$_r$ -- see details at page \pageref{prf:rules-preserved}. One can easily show that they also preserve the semantics:
\[ \interp{\interp{.}^{\repiot \to ZX_r}} = \interp{.}= \interp{\interp{.}^{ZX_r \to \repiot}}\]
Moreover, for any \repiot-diagram $D$: $\repiot \vdash D = \interp{\interp{D}^{\repiot \to ZX_r}}^{ZX_r \to \repiot}$ -- see details at page \pageref{sec:pi_2-double-interpretation}.

Now, let $D_1$ and $D_2$ be two \repiot-diagrams such that $\interp{D_1}=\interp{D_2}$. The two interpretations preserve the semantics, so: 
$\interp{\interp{D_1}^{\repiot \to ZX_r}} = \interp{\interp{D_2}^{\repiot \to ZX_r}}$.\\
Since ZX$_r$ is complete \cite{pivoting}, $ZX_r \vdash \interp{D_1}^{\repiot \to ZX_r} = \interp{D_2}^{\repiot \to ZX_r}$.\\
Moreover, \repiot proves all the equalities of the ZX$_r$, so:\\$\repiot \vdash \interp{\interp{D_1}^{\repiot \to ZX_r}}^{ZX_r \to \repiot} = \interp{\interp{D_2}^{\repiot \to ZX_r}}^{ZX_r \to \repiot}$.\\
Finally, since \repiot proves that the composition of the two interpretations is the identity,
\[ \repiot \vdash D_1 = \interp{\interp{D_1}^{\repiot \to ZX_r}}^{ZX_r \to \repiot} = \interp{\interp{D_2}^{\repiot \to ZX_r}}^{ZX_r \to \repiot} = D_2 \]
which proves the completeness of \repiot.
\end{proof}

\section{From \reca to ZX-Calculus and back}
\label{sec:y-to-zx}

In this section we will explain how to transform diagrams of the
Y-calculus into diagrams of the ZX calculus in a manner that preserves
the semantics -- the diagrams represent the same matrices -- and the
proofs -- if an equality of diagrams is provable in the Y-calculus, the
equality of their images is provable in the ZX-calculus --, and we will
provide a transformation in the reverse direction.

Transforming diagrams from the Y-calculus to the ZX-calculus is easy,
as the real box is representable in the ZX-Calculus. Indeed, we can show that:
\[\interp{\begin{tikzpicture}
	\begin{pgfonlayer}{nodelayer}
		\node [style=bx] (0) at (0, -0) {$\alpha$};
		\node [style=none] (1) at (0, 0.5) {};
		\node [style=none] (2) at (0, -0.5) {};
	\end{pgfonlayer}
	\begin{pgfonlayer}{edgelayer}
		\draw (1.center) to (2.center);
	\end{pgfonlayer}
\end{tikzpicture}} = \begin{pmatrix}
\cos(\alpha/2) & -\sin(\alpha/2) \\
\sin(\alpha/2) & \cos(\alpha/2)
\end{pmatrix} =
\begin{pmatrix}1&0\\0&i\end{pmatrix}
\begin{pmatrix}
\cos(\alpha/2) & -i\sin(\alpha/2) \\
-i\sin(\alpha/2) & \cos(\alpha/2)
\end{pmatrix}
\begin{pmatrix}1&0\\0&-i\end{pmatrix} = \interp{\begin{tikzpicture}
	\begin{pgfonlayer}{nodelayer}
		\node [style=rn] (0) at (-0.5, -0) {$\alpha$};
		\node [style=gn] (1) at (-0.5, -0.75) {$\frac{\pi}{2}$};
		\node [style=gn] (2) at (-0.5, 0.75) {$\frac{-\pi}{2}$};
		\node [style=none] (3) at (-0.5, 1.25) {};
		\node [style=none] (4) at (-0.5, -1.25) {};
		\node [style=rn] (5) at (0.7500001, 0.5) {};
		\node [style=rn] (6) at (0.2499999, 0.5) {$\pi$};
		\node [style=gn] (7) at (0.7500001, -0) {};
		\node [style=gn] (8) at (0.2499999, -0.2499999) {$\frac{-\alpha}{2}$};
	\end{pgfonlayer}
	\begin{pgfonlayer}{edgelayer}
		\draw (3.center) to (4.center);
		\draw [bend right=45, looseness=1.00] (5) to (7);
		\draw [bend right=45, looseness=1.00] (7) to (5);
		\draw (5) to (7);
		\draw (6) to (8);
	\end{pgfonlayer}
\end{tikzpicture}}\]
Hence:
\[\interp{.}^{\rec \to ZX}: \left\lbrace ~~ \begin{array}{l}
D_1\circ D_2 \mapsto \interp{D_1}^{\rec \to ZX} \circ \interp{D_2}^{\rec \to ZX}\\
D_1\otimes D_2 \mapsto \interp{D_1}^{\rec \to ZX} \otimes \interp{D_2}^{\rec \to ZX}\\\\
\begin{tikzpicture}
	\begin{pgfonlayer}{nodelayer}
		\node [style=bx] (0) at (0, -0) {$\alpha$};
		\node [style=none] (1) at (0, 0.5) {};
		\node [style=none] (2) at (0, -0.5) {};
	\end{pgfonlayer}
	\begin{pgfonlayer}{edgelayer}
		\draw (1.center) to (2.center);
	\end{pgfonlayer}
\end{tikzpicture}~\mapsto~\begin{tikzpicture}
	\begin{pgfonlayer}{nodelayer}
		\node [style=rn] (0) at (-0.5, -0) {$\alpha$};
		\node [style=gn] (1) at (-0.5, -0.75) {$\frac{\pi}{2}$};
		\node [style=gn] (2) at (-0.5, 0.75) {$\frac{-\pi}{2}$};
		\node [style=none] (3) at (-0.5, 1.25) {};
		\node [style=none] (4) at (-0.5, -1.25) {};
		\node [style=rn] (5) at (0.7500001, 0.5) {};
		\node [style=rn] (6) at (0.2499999, 0.5) {$\pi$};
		\node [style=gn] (7) at (0.7500001, -0) {};
		\node [style=gn] (8) at (0.2499999, -0.2499999) {$\frac{-\alpha}{2}$};
	\end{pgfonlayer}
	\begin{pgfonlayer}{edgelayer}
		\draw (3.center) to (4.center);
		\draw [bend right=45, looseness=1.00] (5) to (7);
		\draw [bend right=45, looseness=1.00] (7) to (5);
		\draw (5) to (7);
		\draw (6) to (8);
	\end{pgfonlayer}
\end{tikzpicture}\\\\
Id~~\text{otherwise}
\end{array}\right.\]
is an application from the \reca to the ZX-Calculus that preserves the semantics.

\begin{proposition}
\label{prop:y-to-zx-interpretation}
The interpretation $\interp{.}^{\rec \to ZX}$ preserves all the rules of the \reca, so:
\[\forall D_1, D_2,\quad(\rec \vdash D_1=D_2)\implies \left(ZX\vdash \interp{D_1}^{\rec \to ZX}=\interp{D_2}^{\rec \to ZX}\right)\]
\end{proposition}

\begin{proof}
In appendix at page \pageref{prf:y-to-zx-interpretation}
\end{proof}
Note that if the  diagram of the Y-calculus has angles in a fragment $\pi/k$
then the corresponding diagram of the ZX-calculus has angles (actually
scalars) in the fragment $\pi/2k$.

Going in the other direction is harder as, evidently, a matrix with
complex coefficients  is usually not a matrix with real coefficients.
There is however a way to palliate the problem by converting a complex matrix of
size $p \times p$  to a real matrix of size $2p \times 2p$,
essentially using the following coding of complex numbers into $2
\times 2$ real matrices:

\[ a+ib \mapsto \begin{pmatrix}
 a & b \\
 -b & a \\
\end{pmatrix}
\]

\label{sec:zx-to-y}

Doing so essentially adds one wire to the diagram so that a diagram $n
\rightarrow m$ will be transformed into a diagram $n+1 \rightarrow m+1$.
This leads to difficulties in the handling of the spatial composition.

Specifically, the interpretation is as follows:

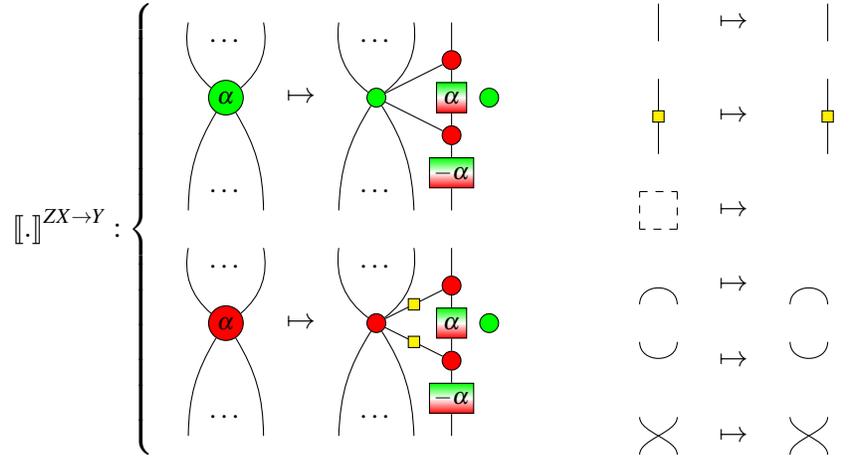
\[ \interp{.}^{ZX\to\rec}: \left\lbrace\quad \begin{tikzpicture}
	\begin{pgfonlayer}{nodelayer}
		\node [style=gn] (0) at (-2.25, 1.75) {};
		\node [style=bx] (1) at (-1.25, 1.75) {$\alpha$};
		\node [style=rn] (2) at (-1.25, 2.25) {};
		\node [style=rn] (3) at (-1.25, 1.25) {};
		\node [style=bx] (4) at (-1.25, 0.7500001) {$-\alpha$};
		\node [style=none] (5) at (-2.75, 2.75) {};
		\node [style=none] (6) at (-1.25, 2.75) {};
		\node [style=none] (7) at (-1.25, 0.25) {};
		\node [style=none] (8) at (-2.75, 0.25) {};
		\node [style=none] (9) at (-1.75, 2.75) {};
		\node [style=none] (10) at (-1.75, 0.25) {};
		\node [style=none] (11) at (-2.25, 2.5) {$\cdots$};
		\node [style=none] (12) at (-2.25, 0.5000001) {$\cdots$};
		\node [style=gn] (13) at (-0.7500001, 1.75) {};
		\node [style=none] (14) at (-3.75, 0.25) {};
		\node [style=none] (15) at (-4.75, 0.25) {};
		\node [style=none] (16) at (-4.25, 2.5) {$\cdots$};
		\node [style=none] (17) at (-4.75, 2.75) {};
		\node [style=gn] (18) at (-4.25, 1.75) {$~\alpha~$};
		\node [style=none] (19) at (-3.75, 2.75) {};
		\node [style=none] (20) at (-4.25, 0.5000001) {$\cdots$};
		\node [style=none] (21) at (-3.25, 1.75) {$\mapsto$};
		\node [style=none] (22) at (-3.25, -1.25) {$\mapsto$};
		\node [style=gn] (23) at (-0.7500001, -1.25) {};
		\node [style=none] (24) at (-1.75, -2.75) {};
		\node [style=none] (25) at (-4.75, -2.75) {};
		\node [style=rn] (26) at (-4.25, -1.25) {$~\alpha~$};
		\node [style=bx] (27) at (-1.25, -1.25) {$\alpha$};
		\node [style=none] (28) at (-2.75, -2.75) {};
		\node [style=none] (29) at (-3.75, -0.25) {};
		\node [style=rn] (30) at (-1.25, -1.75) {};
		\node [style=none] (31) at (-1.25, -0.25) {};
		\node [style=none] (32) at (-4.75, -0.25) {};
		\node [style=rn] (33) at (-2.25, -1.25) {};
		\node [style=none] (34) at (-1.25, -2.75) {};
		\node [style=rn] (35) at (-1.25, -0.7500001) {};
		\node [style=none] (36) at (-2.25, -2.5) {$\cdots$};
		\node [style=none] (37) at (-4.25, -0.5000001) {$\cdots$};
		\node [style=none] (38) at (-2.75, -0.25) {};
		\node [style=bx] (39) at (-1.25, -2.25) {$-\alpha$};
		\node [style=none] (40) at (-2.25, -0.5000001) {$\cdots$};
		\node [style=none] (41) at (-1.75, -0.25) {};
		\node [style=none] (42) at (-4.25, -2.5) {$\cdots$};
		\node [style=none] (43) at (-3.75, -2.75) {};
		\node [style={{H box}}] (44) at (-1.75, -1) {};
		\node [style={{H box}}] (45) at (-1.75, -1.5) {};
		\node [style=none] (46) at (1.5, 2) {};
		\node [style=none] (47) at (1.5, 1) {};
		\node [style={{H box}}] (48) at (1.5, 1.5) {};
		\node [style=none] (49) at (3.75, 1) {};
		\node [style={{H box}}] (50) at (3.75, 1.5) {};
		\node [style=none] (51) at (3.75, 2) {};
		\node [style=none] (52) at (4.25, 1) {};
		\node [style=none] (53) at (4.25, 2) {};
		\node [style=none] (54) at (2.5, 1.5) {$\mapsto$};
		\node [style=none] (55) at (1.25, -1) {};
		\node [style=none] (56) at (4.25, -0.5000001) {};
		\node [style=none] (57) at (1.75, -1) {};
		\node [style=none] (58) at (4.25, -1) {};
		\node [style=none] (59) at (2.5, -0.7500001) {$\mapsto$};
		\node [style=none] (60) at (3.75, -1) {};
		\node [style=none] (61) at (3.25, -1) {};
		\node [style=none] (62) at (3.75, -1.5) {};
		\node [style=none] (63) at (2.5, -1.75) {$\mapsto$};
		\node [style=none] (64) at (1.25, -1.5) {};
		\node [style=none] (65) at (4.25, -1.5) {};
		\node [style=none] (66) at (3.25, -1.5) {};
		\node [style=none] (67) at (4.25, -2) {};
		\node [style=none] (68) at (1.75, -1.5) {};
		\node [style=none] (69) at (2.5, -2.75) {$\mapsto$};
		\node [style=none] (70) at (1.75, -3) {};
		\node [style=none] (71) at (4.25, -2.5) {};
		\node [style=none] (72) at (4.25, -3) {};
		\node [style=none] (73) at (1.25, -3) {};
		\node [style=none] (74) at (1.75, -2.5) {};
		\node [style=none] (75) at (1.25, -2.5) {};
		\node [style=none] (76) at (3.25, -2.5) {};
		\node [style=none] (77) at (3.25, -3) {};
		\node [style=none] (78) at (3.75, -3) {};
		\node [style=none] (79) at (3.75, -2.5) {};
		\node [style=none] (80) at (1.75, -0) {};
		\node [style=none] (81) at (2.5, 0.25) {$\mapsto$};
		\node [style=none] (82) at (1.25, 0.5000001) {};
		\node [style=none] (83) at (4.25, 0.5000001) {};
		\node [style=none] (84) at (1.25, -0) {};
		\node [style=none] (85) at (1.75, 0.5000001) {};
		\node [style=none] (86) at (4.25, -0) {};
		\node [style=none] (87) at (3.75, 3) {};
		\node [style=none] (88) at (3.75, 2.5) {};
		\node [style=none] (89) at (1.5, 3) {};
		\node [style=none] (90) at (4.25, 3) {};
		\node [style=none] (91) at (1.5, 2.5) {};
		\node [style=none] (92) at (4.25, 2.5) {};
		\node [style=none] (93) at (2.5, 2.75) {$\mapsto$};
	\end{pgfonlayer}
	\begin{pgfonlayer}{edgelayer}
		\draw [bend right, looseness=1.00] (5.center) to (0);
		\draw [bend right=15, looseness=1.00] (0) to (8.center);
		\draw (6.center) to (2);
		\draw (2) to (0);
		\draw (0) to (3);
		\draw (3) to (4);
		\draw (4) to (7.center);
		\draw (2) to (1);
		\draw (1) to (3);
		\draw [style=none, bend left, looseness=1.00] (9.center) to (0);
		\draw [style=none, bend left=15, looseness=1.00] (0) to (10.center);
		\draw [bend right, looseness=1.00] (17.center) to (18);
		\draw [bend right=15, looseness=1.00] (18) to (15.center);
		\draw [style=none, bend left, looseness=1.00] (19.center) to (18);
		\draw [style=none, bend left=15, looseness=1.00] (18) to (14.center);
		\draw [bend right, looseness=1.00] (38.center) to (33);
		\draw [bend right=15, looseness=1.00] (33) to (28.center);
		\draw (31.center) to (35);
		\draw (35) to (33);
		\draw (33) to (30);
		\draw (30) to (39);
		\draw (39) to (34.center);
		\draw (35) to (27);
		\draw (27) to (30);
		\draw [style=none, bend left, looseness=1.00] (41.center) to (33);
		\draw [style=none, bend left=15, looseness=1.00] (33) to (24.center);
		\draw [bend right, looseness=1.00] (32.center) to (26);
		\draw [bend right=15, looseness=1.00] (26) to (25.center);
		\draw [style=none, bend left, looseness=1.00] (29.center) to (26);
		\draw [style=none, bend left=15, looseness=1.00] (26) to (43.center);
		\draw (46.center) to (47.center);
		\draw (51.center) to (49.center);
		\draw (53.center) to (52.center);
		\draw [bend left=90, looseness=1.50] (55.center) to (57.center);
		\draw (56.center) to (58.center);
		\draw [bend left=90, looseness=1.50] (61.center) to (60.center);
		\draw [bend left=90, looseness=1.50] (68.center) to (64.center);
		\draw (65.center) to (67.center);
		\draw [bend left=90, looseness=1.50] (62.center) to (66.center);
		\draw (71.center) to (72.center);
		\draw [in=90, out=-90, looseness=0.75] (75.center) to (70.center);
		\draw [in=90, out=-90, looseness=0.75] (74.center) to (73.center);
		\draw [in=90, out=-90, looseness=0.75] (76.center) to (78.center);
		\draw [in=90, out=-90, looseness=0.75] (79.center) to (77.center);
		\draw (83.center) to (86.center);
		\draw [style=dashed] (82.center) to (85.center);
		\draw [style=dashed] (85.center) to (80.center);
		\draw [style=dashed] (80.center) to (84.center);
		\draw [style=dashed] (84.center) to (82.center);
		\draw (89.center) to (91.center);
		\draw (87.center) to (88.center);
		\draw (90.center) to (92.center);
	\end{pgfonlayer}
\end{tikzpicture} \right. \]

\noindent\textbf{Sequential Composition:} The interpretation is a morphism for $\circ$:
\[ D_1\circ D_2 \mapsto \interp{D_1}^{ZX\to \rec} \circ \interp{D_2}^{ZX\to \rec}\]

\noindent\textbf{Spatial Composition:}
The interpretation changes the way two side by side diagrams are represented: $\interp{.\otimes .}^{ZX\to\rec}\neq\interp{.}^{ZX\to\rec}\otimes\interp{.}^{ZX\to\rec}$.
Instead, the two interpreted diagrams share the last wire, called \textit{control wire}. Given $D_n$ a ZX-diagram with $n$ inputs and $n'$ outputs, and $D_m$ a ZX-diagram with $m$ inputs, the interpretation of $D_n$ side-by-side with $D_m$ is:
\[ \fit{$\interp{D_n\otimes D_m}^{ZX\to\rec} = \left( \mathbb{I}^{\otimes n'}\otimes\interp{D_m}^{ZX\to\rec}\right)\circ\left(\nmswapI{m}{n'}\right)\circ \left(\mathbb{I}^{\otimes m}\otimes\interp{D_n}^{ZX\to\rec}\right)\circ\left(\nmswapI{n}{m}\right)$} \]
Assuming the interpretation of $D$ is written this way:
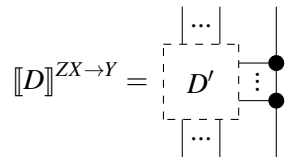
\[ \interp{D}^{ZX\to\rec}=~\begin{tikzpicture}
	\begin{pgfonlayer}{nodelayer}
		\node [style=none] (0) at (-0.5000001, 0.5000001) {};
		\node [style=none] (1) at (1, -0.9999999) {};
		\node [style=none] (2) at (-0.5000001, -0.5000001) {};
		\node [style=none] (3) at (0.2499997, -0.9999999) {};
		\node [style=dot] (4) at (1, -0.2500001) {};
		\node [style=none] (5) at (0.5000001, 0.2500001) {};
		\node [style=none] (6) at (0.2499997, 0.5000001) {};
		\node [style=none] (7) at (0.5000001, 0.5000001) {};
		\node [style=none] (8) at (-0.2499997, -0.5000001) {};
		\node [style=none] (9) at (0.7499998, -0) {\rotatebox[origin=c]{90}{...}};
		\node [style=none] (10) at (0.5000001, -0.2500001) {};
		\node [style=none] (11) at (-0.2499997, 0.5000001) {};
		\node [style=none] (12) at (0.2499997, -0.5000001) {};
		\node [style=none] (13) at (0, -0.7499998) {...};
		\node [style=none] (14) at (0.5000001, -0.5000001) {};
		\node [style=none] (15) at (-0.2499997, -0.9999999) {};
		\node [style=dot] (16) at (1, 0.2500001) {};
		\node [style=none] (17) at (0, -0) {$D'$};
		\node [style=none] (18) at (1, 0.9999999) {};
		\node [style=none] (19) at (-0.2499997, 0.9999999) {};
		\node [style=none] (20) at (0.2499997, 0.9999999) {};
		\node [style=none] (21) at (0, 0.7499998) {...};
	\end{pgfonlayer}
	\begin{pgfonlayer}{edgelayer}
		\draw [style=dashed] (0.center) to (7.center);
		\draw [style=dashed] (7.center) to (14.center);
		\draw [style=dashed] (14.center) to (2.center);
		\draw [style=dashed] (2.center) to (0.center);
		\draw [style=none] (8.center) to (15.center);
		\draw [style=none] (12.center) to (3.center);
		\draw [style=none] (5.center) to (16);
		\draw [style=none] (10.center) to (4);
		\draw [style=none] (18.center) to (1.center);
		\draw [style=none] (19.center) to (11.center);
		\draw [style=none] (20.center) to (6.center);
	\end{pgfonlayer}
\end{tikzpicture} \]
We can roughly see the spatial composition as:
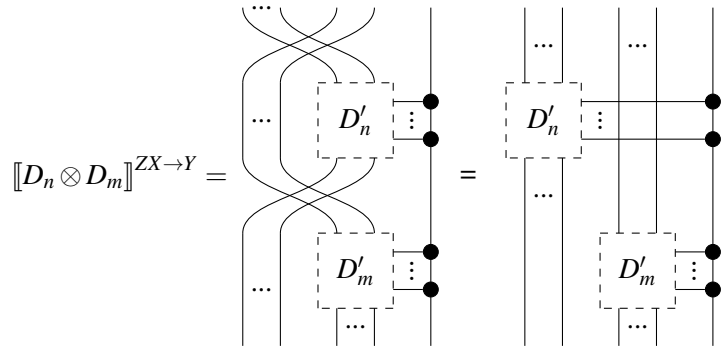
\[ \interp{D_n\otimes D_m}^{ZX\to\rec}=~\begin{tikzpicture}
	\begin{pgfonlayer}{nodelayer}
		\node [style=none] (0) at (-2, -0.7499998) {};
		\node [style=none] (1) at (-1, -0.7499998) {};
		\node [style=none] (2) at (-2, -1.75) {};
		\node [style=none] (3) at (-1, -1.75) {};
		\node [style=none] (4) at (-1.75, -0.7499998) {};
		\node [style=none] (5) at (-1.25, -0.7499998) {};
		\node [style=none] (6) at (-1.75, -1.75) {};
		\node [style=none] (7) at (-1.25, -1.75) {};
		\node [style=none] (8) at (-1, -0.9999999) {};
		\node [style=none] (9) at (-1, -1.5) {};
		\node [style=none] (10) at (-1.5, -1.25) {$D_m'$};
		\node [style=none] (11) at (-2.5, -0.7499998) {};
		\node [style=none] (12) at (-3, -0.7499998) {};
		\node [style=none] (13) at (-1.75, -2.25) {};
		\node [style=none] (14) at (-1.25, -2.25) {};
		\node [style=none] (15) at (-3, -2.25) {};
		\node [style=none] (16) at (-2.5, -2.25) {};
		\node [style=none] (17) at (-0.5000001, -2.25) {};
		\node [style=none] (18) at (-1, 0.5000001) {};
		\node [style=none] (19) at (-1.75, 1.25) {};
		\node [style=none] (20) at (-1, 0.9999999) {};
		\node [style=none] (21) at (-2, 1.25) {};
		\node [style=none] (22) at (-1, 0.2500001) {};
		\node [style=none] (23) at (-1.25, 1.25) {};
		\node [style=none] (24) at (-1.5, 0.7499998) {$D_n'$};
		\node [style=none] (25) at (-2.5, 0.2500001) {};
		\node [style=none] (26) at (-1.25, 0.2500001) {};
		\node [style=none] (27) at (-1, 1.25) {};
		\node [style=none] (28) at (-2, 0.2500001) {};
		\node [style=none] (29) at (-1.75, 0.2500001) {};
		\node [style=none] (30) at (-3, 0.2500001) {};
		\node [style=none] (31) at (-2.5, 2.25) {};
		\node [style=none] (32) at (-1.75, 2.25) {};
		\node [style=none] (33) at (-1.25, 2.25) {};
		\node [style=none] (34) at (-0.5000001, 2.25) {};
		\node [style=none] (35) at (-3, 2.25) {};
		\node [style=none] (36) at (-3, 1.25) {};
		\node [style=none] (37) at (-2.5, 1.25) {};
		\node [style=dot] (38) at (-0.5000001, -0.9999999) {};
		\node [style=dot] (39) at (-0.5000001, -1.5) {};
		\node [style=dot] (40) at (-0.5000001, 0.9999999) {};
		\node [style=dot] (41) at (-0.5000001, 0.5000001) {};
		\node [style=none] (42) at (0, -0) {=};
		\node [style=none] (43) at (2.5, -0.7499998) {};
		\node [style=none] (44) at (1.5, 0.2500001) {};
		\node [style=dot] (45) at (3.25, 0.5000001) {};
		\node [style=none] (46) at (1.5, 1.25) {};
		\node [style=none] (47) at (0.7499998, 1.25) {};
		\node [style=none] (48) at (0.5000001, 0.2500001) {};
		\node [style=none] (49) at (1.25, 0.2500001) {};
		\node [style=dot] (50) at (3.25, -0.9999999) {};
		\node [style=none] (51) at (1, 0.7499998) {$D_n'$};
		\node [style=none] (52) at (2.75, -1.5) {};
		\node [style=none] (53) at (1.75, -0.7499998) {};
		\node [style=none] (54) at (3.25, 2.25) {};
		\node [style=none] (55) at (2, 2.25) {};
		\node [style=none] (56) at (0.7499998, -2.25) {};
		\node [style=none] (57) at (1.25, 2.25) {};
		\node [style=none] (58) at (2.5, -1.75) {};
		\node [style=dot] (59) at (3.25, 0.9999999) {};
		\node [style=none] (60) at (2.75, -0.7499998) {};
		\node [style=none] (61) at (0.7499998, 2.25) {};
		\node [style=none] (62) at (2.25, -1.25) {$D_m'$};
		\node [style=none] (63) at (0.7499998, 0.2500001) {};
		\node [style=none] (64) at (2.75, -0.9999999) {};
		\node [style=none] (65) at (2, -0.7499998) {};
		\node [style=none] (66) at (2, -2.25) {};
		\node [style=none] (67) at (2.5, -2.25) {};
		\node [style=none] (68) at (2.75, -1.75) {};
		\node [style=none] (69) at (2, -1.75) {};
		\node [style=none] (70) at (3.25, -2.25) {};
		\node [style=dot] (71) at (3.25, -1.5) {};
		\node [style=none] (72) at (1.5, 0.9999999) {};
		\node [style=none] (73) at (2.5, 2.25) {};
		\node [style=none] (74) at (1.75, -1.75) {};
		\node [style=none] (75) at (1.25, -2.25) {};
		\node [style=none] (76) at (0.5000001, 1.25) {};
		\node [style=none] (77) at (1.25, 1.25) {};
		\node [style=none] (78) at (1.5, 0.5000001) {};
		\node [style=none] (79) at (-0.7499998, 0.7499998) {\rotatebox[origin=c]{90}{...}};
		\node [style=none] (80) at (-0.7499998, -1.25) {\rotatebox[origin=c]{90}{...}};
		\node [style=none] (81) at (3, -1.25) {\rotatebox[origin=c]{90}{...}};
		\node [style=none] (82) at (1.75, 0.7499998) {\rotatebox[origin=c]{90}{...}};
		\node [style=none] (83) at (-2.75, 0.7499998) {...};
		\node [style=none] (84) at (-2.75, -1.5) {...};
		\node [style=none] (85) at (-1.5, -2) {...};
		\node [style=none] (86) at (-2.75, 2.25) {...};
		\node [style=none] (87) at (1, 1.75) {...};
		\node [style=none] (88) at (1, -0.2500001) {...};
		\node [style=none] (89) at (2.25, 1.75) {...};
		\node [style=none] (90) at (2.25, -2) {...};
	\end{pgfonlayer}
	\begin{pgfonlayer}{edgelayer}
		\draw [style=dashed] (0.center) to (1.center);
		\draw [style=dashed] (1.center) to (3.center);
		\draw [style=dashed] (3.center) to (2.center);
		\draw [style=dashed] (2.center) to (0.center);
		\draw [style=dashed] (21.center) to (27.center);
		\draw [style=dashed] (27.center) to (22.center);
		\draw [style=dashed] (22.center) to (28.center);
		\draw [style=dashed] (28.center) to (21.center);
		\draw [style=none] (34.center) to (17.center);
		\draw [style=none] (36.center) to (30.center);
		\draw [style=none] (37.center) to (25.center);
		\draw [style=none] (12.center) to (15.center);
		\draw [style=none] (11.center) to (16.center);
		\draw [style=none] (6.center) to (13.center);
		\draw [style=none] (7.center) to (14.center);
		\draw [style=none] (18.center) to (41);
		\draw [style=none] (20.center) to (40);
		\draw [style=none] (8.center) to (38);
		\draw [style=none] (9.center) to (39);
		\draw [style=none, in=-90, out=90, looseness=0.75] (4.center) to (30.center);
		\draw [style=none, in=90, out=-90, looseness=0.75] (25.center) to (5.center);
		\draw [style=none, in=90, out=-90, looseness=0.75] (29.center) to (12.center);
		\draw [style=none, in=-90, out=90, looseness=0.75] (11.center) to (26.center);
		\draw [style=none, in=-90, out=90, looseness=0.75] (19.center) to (35.center);
		\draw [style=none, in=90, out=-90, looseness=0.75] (31.center) to (23.center);
		\draw [style=none, in=90, out=-90, looseness=0.75] (32.center) to (36.center);
		\draw [style=none, in=-90, out=90, looseness=0.75] (37.center) to (33.center);
		\draw [style=dashed] (53.center) to (60.center);
		\draw [style=dashed] (60.center) to (68.center);
		\draw [style=dashed] (68.center) to (74.center);
		\draw [style=dashed] (74.center) to (53.center);
		\draw [style=dashed] (76.center) to (46.center);
		\draw [style=dashed] (46.center) to (44.center);
		\draw [style=dashed] (44.center) to (48.center);
		\draw [style=dashed] (48.center) to (76.center);
		\draw [style=none] (54.center) to (70.center);
		\draw [style=none] (69.center) to (66.center);
		\draw [style=none] (58.center) to (67.center);
		\draw [style=none] (78.center) to (45);
		\draw [style=none] (72.center) to (59);
		\draw [style=none] (64.center) to (50);
		\draw [style=none] (52.center) to (71);
		\draw [style=none, in=-90, out=90, looseness=0.75] (47.center) to (61.center);
		\draw [style=none, in=90, out=-90, looseness=0.75] (57.center) to (77.center);
		\draw [style=none] (63.center) to (56.center);
		\draw [style=none] (49.center) to (75.center);
		\draw [style=none] (55.center) to (65.center);
		\draw [style=none] (73.center) to (43.center);
	\end{pgfonlayer}
\end{tikzpicture} \]

\begin{lemma}
\label{lem:commutation-on-control-wire}
All the subdiagrams generated by the interpretation can commute on the control wire.
\end{lemma}
\begin{proof}
In appendix at page \pageref{prf:commutation-on-control-wire}.
\end{proof}

\noindent
Now with this result, we can show:
\begin{itemize}
\item $\interp{(A_1\otimes B_1)\circ(A_2\otimes B_2)}^{ZX\to\rec} = \interp{(A_1\circ A2)\otimes(B_1\circ B2)}^{ZX\to\rec}$ if the number of outputs of $A_2$ (resp. $B_2$) corresponds to the number of inputs of $A_1$ (resp. $B_1$)
\item $\interp{(D_1\otimes D_2)\otimes D_3}^{ZX\to\rec}=\interp{D_1\otimes (D_2\otimes D_3)}^{ZX\to\rec}$
\item $\interp{e\otimes D}^{ZX\to\rec} = \interp{D \otimes e}^{ZX\to\rec}= \interp{D}^{ZX\to\rec}$
\item $\interp{(D_1\otimes D_2)\circ \sigma}^{ZX\to\rec}=\interp{\sigma \circ(D_2\otimes D_1)}^{ZX\to\rec}$ for any 1-input/1-output diagrams $D_1$ and $D_2$
\item Any topological property of the ZX-Calculus is preserved.
\end{itemize}

\begin{proposition}
\label{prop:zx-to-y-interpretation}
All the rules of the ZX-Calculus -- see figure \ref{fig:ZX_rules} -- are preserved with the interpretation $\interp{.}^{ZX\to\rec}$:
\[\forall D_1, D_2,\quad(ZX \vdash D_1=D_2)\implies \left(\rec\vdash \interp{D_1}^{ZX\to\rec}=\interp{D_2}^{ZX\to\rec}\right)\]
\end{proposition}
\begin{proof}
In appendix at page \pageref{prf:zx-to-y-interpretation}.
\end{proof}

\begin{proposition}
\label{prop:zx2y-decomp}
For any diagram $D$:
\[
\interp{\interp{D}^{ZX\to\rec}} = \Re(\interp{D})\otimes I_2 + \Im(\interp{D}) \otimes \iY
\]
\end{proposition}

\begin{proof}
In appendix at page \pageref{prf:zx2y-decomp}
\end{proof}

The two interpretations above show that the rules of the Y-calculus we give are
the right ones: they are able to prove all rules of the ZX-calculus, and they
are all provable in the ZX-calculus.
We can make this more formal:
\begin{proposition}
\label{prop:inverses}
One can retrieve the initial diagram after the composition of both interpretations:
\[\forall D\in\rec,~~ \rec\vdash\left(\begin{tikzpicture}
	\begin{pgfonlayer}{nodelayer}
		\node [style=rn] (0) at (0.5000001, -0.2500001) {};
		\node [style=none] (1) at (0.5000001, 0.2500001) {};
		\node [style=none] (2) at (0, 0.2500001) {};
		\node [style=none] (3) at (0, -0.2500001) {};
		\node [style=none] (4) at (-1, 0.2500001) {};
		\node [style=none] (5) at (-1, -0.2500001) {};
		\node [style=none] (6) at (-0.5000001, -0) {$~\cdots~$};
		\node [style=rn] (7) at (1, 0.2500001) {};
		\node [style=gn] (8) at (1, -0.2500001) {};
	\end{pgfonlayer}
	\begin{pgfonlayer}{edgelayer}
		\draw (4.center) to (5.center);
		\draw (2.center) to (3.center);
		\draw (1.center) to (0);
		\draw [bend right=45, looseness=1.00] (7) to (8);
		\draw [bend right=45, looseness=1.00] (8) to (7);
		\draw (7) to (8);
	\end{pgfonlayer}
\end{tikzpicture}\right)\circ\interp{\interp{D}^{\rec\to ZX}}^{ZX\to\rec}\circ\left(\begin{tikzpicture}
	\begin{pgfonlayer}{nodelayer}
		\node [style=rn] (0) at (0.5000001, 0.2500001) {};
		\node [style=none] (1) at (0.5000001, -0.2500001) {};
		\node [style=none] (2) at (0, 0.2500001) {};
		\node [style=none] (3) at (0, -0.2500001) {};
		\node [style=none] (4) at (-1, 0.2500001) {};
		\node [style=none] (5) at (-1, -0.2500001) {};
		\node [style=none] (6) at (-0.5000001, -0) {$~\cdots~$};
		\node [style=rn] (7) at (1, 0.2500001) {};
		\node [style=gn] (8) at (1, -0.2500001) {};
	\end{pgfonlayer}
	\begin{pgfonlayer}{edgelayer}
		\draw (4.center) to (5.center);
		\draw (2.center) to (3.center);
		\draw (1.center) to (0);
		\draw [bend right=45, looseness=1.00] (7) to (8);
		\draw [bend right=45, looseness=1.00] (8) to (7);
		\draw (7) to (8);
	\end{pgfonlayer}
\end{tikzpicture}\right)=D \]
\[\forall D\in ZX,~~ ZX\vdash\left(\begin{tikzpicture}
	\begin{pgfonlayer}{nodelayer}
		\node [style=rn] (0) at (0.5000001, -0.2500001) {};
		\node [style=none] (1) at (0.5000001, 0.2500001) {};
		\node [style=none] (2) at (0, 0.2500001) {};
		\node [style=none] (3) at (0, -0.2500001) {};
		\node [style=none] (4) at (-1, 0.2500001) {};
		\node [style=none] (5) at (-1, -0.2500001) {};
		\node [style=none] (6) at (-0.5000001, -0) {$~\cdots~$};
		\node [style=rn] (7) at (1, 0.2500001) {};
		\node [style=gn] (8) at (1, -0.2500001) {};
	\end{pgfonlayer}
	\begin{pgfonlayer}{edgelayer}
		\draw (4.center) to (5.center);
		\draw (2.center) to (3.center);
		\draw (1.center) to (0);
		\draw [bend right=45, looseness=1.00] (7) to (8);
		\draw [bend right=45, looseness=1.00] (8) to (7);
		\draw (7) to (8);
	\end{pgfonlayer}
\end{tikzpicture}\right)\circ\interp{\interp{D}^{ZX\to\rec}}^{\rec\to ZX}\circ\left(\begin{tikzpicture}
	\begin{pgfonlayer}{nodelayer}
		\node [style=none] (0) at (0.5000001, -0.2500001) {};
		\node [style=none] (1) at (0, 0.2500001) {};
		\node [style=none] (2) at (0, -0.2500001) {};
		\node [style=none] (3) at (-1, 0.2500001) {};
		\node [style=none] (4) at (-1, -0.2500001) {};
		\node [style=none] (5) at (-0.5000001, -0) {$~\cdots~$};
		\node [style=gn] (6) at (0.5000001, 0.2500001) {$\frac{\pi}{2}$};
	\end{pgfonlayer}
	\begin{pgfonlayer}{edgelayer}
		\draw (3.center) to (4.center);
		\draw (1.center) to (2.center);
		\draw (6) to (0.center);
	\end{pgfonlayer}
\end{tikzpicture}\right)=D \]
\end{proposition}

\begin{proof}
In appendix at page \pageref{prf:inverses}.
\end{proof}
\begin{corollary}
 If $ZX \vdash \interp{\interp{D_1}^{ZX\to\rec}}^{\rec\to ZX} = 
 \interp{\interp{D_2}^{ZX\to\rec}}^{\rec\to ZX}$ then $ZX \vdash D_1 = D_2$.

 If $Y \vdash \interp{\interp{D_1}^{\rec\to ZX}}^{ZX \to\rec} = 
\interp{\interp{D_2}^{\rec\to ZX}}^{ZX \to\rec} $ then $Y \vdash D_1 = D_2$.
\end{corollary}

As a consequence:

\begin{theorem}
 The ZX-Calculus is complete if and only if the Y-Calculus is complete.  
\end{theorem}
\begin{proof}
 Suppose that the ZX-Calculus is complete.
 Let $D_1, D_2$ be two diagrams of the Y-Calculus s.t. $\interp{D_1} =
 \interp{D_2}$.
 As the interpretation $\interp{\cdot}^{\rec\to ZX}$ preserves semantics,
 $\interp{\interp{D_1}^{\rec\to ZX}} = \interp{\interp{D_2}^{\rec\to ZX}}$.
 As the ZX-Calculus is complete, $ZX \vdash     \interp{D_1}^{\rec\to ZX} =
 \interp{D_2}^{\rec\to ZX}$.
 As the transformation preserves provability,
 $Y \vdash \interp{\interp{D_1}^{\rec\to ZX}}^{ZX \to\rec} =
 \interp{\interp{D_2}^{\rec\to ZX}}^{ZX \to\rec}$.
 Hence $Y \vdash D_1 = D_2$ by the previous corollary.\\
 The other direction follows mutatis mutandis.  
\end{proof}
The result above is only true for the full ZX-Calculus with arbitrary angles:
Starting from a diagram in the Y-Calculus with a angle $\alpha$, 
the interpretation $\interp{.}^{\rec\to ZX}$ might introduce the angle
$\alpha/2$.
There is a way around this problem that we will explain in a subsequent  paper.

To finish, we explain how the two interpretations also explain how to extract
the real and imaginary parts of a ZX-diagram.
\begin{corollary}
Let $D$ be a ZX-diagram, and the interpretation $\interp{.}^{\natural}$ be either $\interp{.}^{ZX\to\rec}$ or $\interp{\interp{.}^{ZX\to\rec}}^{\rec\to ZX}$. Let us define $\Re(D)$ and $\Im(D)$ as follows:
\[ \Re(D) = \left(\begin{tikzpicture}
	\begin{pgfonlayer}{nodelayer}
		\node [style=none] (0) at (-1, 0.25) {};
		\node [style=none] (1) at (-1, -0.25) {};
		\node [style=none] (2) at (-0.5, -0) {...};
		\node [style=none] (3) at (0, 0.25) {};
		\node [style=none] (4) at (0, -0.25) {};
		\node [style=none] (5) at (0.5, 0.25) {};
		\node [style=rn] (6) at (0.5, -0.25) {};
		\node [style=rn] (7) at (1, 0.25) {};
		\node [style=gn] (8) at (1, -0.25) {};
	\end{pgfonlayer}
	\begin{pgfonlayer}{edgelayer}
		\draw [bend right=45, looseness=1.00] (7) to (8);
		\draw [bend right=45, looseness=1.00] (8) to (7);
		\draw (7) to (8);
		\draw (0.center) to (1.center);
		\draw (3.center) to (4.center);
		\draw (5.center) to (6);
	\end{pgfonlayer}
\end{tikzpicture}\right)\circ\interp{D}^{\natural}\circ \left(\begin{tikzpicture}
	\begin{pgfonlayer}{nodelayer}
		\node [style=none] (0) at (-1, 0.25) {};
		\node [style=none] (1) at (-1, -0.25) {};
		\node [style=none] (2) at (-0.5, -0) {...};
		\node [style=none] (3) at (0, 0.25) {};
		\node [style=none] (4) at (0, -0.25) {};
		\node [style=none] (5) at (0.5, -0.25) {};
		\node [style=rn] (6) at (0.5, 0.25) {};
		\node [style=rn] (7) at (1, 0.25) {};
		\node [style=gn] (8) at (1, -0.25) {};
	\end{pgfonlayer}
	\begin{pgfonlayer}{edgelayer}
		\draw [bend right=45, looseness=1.00] (7) to (8);
		\draw [bend right=45, looseness=1.00] (8) to (7);
		\draw (7) to (8);
		\draw (0.center) to (1.center);
		\draw (3.center) to (4.center);
		\draw (5.center) to (6);
	\end{pgfonlayer}
\end{tikzpicture}\right)\]
\[ \Im(D) = \left(\begin{tikzpicture}
	\begin{pgfonlayer}{nodelayer}
		\node [style=none] (0) at (-1, 0.25) {};
		\node [style=none] (1) at (-1, -0.25) {};
		\node [style=none] (2) at (-0.5, -0) {...};
		\node [style=none] (3) at (0, 0.25) {};
		\node [style=none] (4) at (0, -0.25) {};
		\node [style=none] (5) at (0.5, 0.25) {};
		\node [style=rn] (6) at (0.5, -0.25) {};
		\node [style=rn] (7) at (1, 0.25) {};
		\node [style=gn] (8) at (1, -0.25) {};
	\end{pgfonlayer}
	\begin{pgfonlayer}{edgelayer}
		\draw [bend right=45, looseness=1.00] (7) to (8);
		\draw [bend right=45, looseness=1.00] (8) to (7);
		\draw (7) to (8);
		\draw (0.center) to (1.center);
		\draw (3.center) to (4.center);
		\draw (5.center) to (6);
	\end{pgfonlayer}
\end{tikzpicture}\right)\circ\interp{D}^{\natural}\circ \left(\begin{tikzpicture}
	\begin{pgfonlayer}{nodelayer}
		\node [style=none] (0) at (-1, 0.25) {};
		\node [style=none] (1) at (-1, -0.25) {};
		\node [style=none] (2) at (-0.5, -0) {...};
		\node [style=none] (3) at (0, 0.25) {};
		\node [style=none] (4) at (0, -0.25) {};
		\node [style=none] (5) at (0.5, -0.25) {};
		\node [style=rn] (6) at (0.5, 0.25) {$\pi$};
		\node [style=rn] (7) at (1, 0.25) {};
		\node [style=gn] (8) at (1, -0.25) {};
	\end{pgfonlayer}
	\begin{pgfonlayer}{edgelayer}
		\draw [bend right=45, looseness=1.00] (7) to (8);
		\draw [bend right=45, looseness=1.00] (8) to (7);
		\draw (7) to (8);
		\draw (0.center) to (1.center);
		\draw (3.center) to (4.center);
		\draw (5.center) to (6);
	\end{pgfonlayer}
\end{tikzpicture}\right)\]
Then $\interp{\Re(D)} = \Re(\interp{D})$ and $\interp{\Im(D)} = \Im(\interp{D})$
\end{corollary}
\begin{proof} Let $A$ and $B$ be two real matrices such that $\interp{D}=A+iB$.
\begin{align*}
\left\llbracket\left(\begin{tikzpicture}
	\begin{pgfonlayer}{nodelayer}
		\node [style=none] (0) at (-1, 0.25) {};
		\node [style=none] (1) at (-1, -0.25) {};
		\node [style=none] (2) at (-0.5, -0) {...};
		\node [style=none] (3) at (0, 0.25) {};
		\node [style=none] (4) at (0, -0.25) {};
		\node [style=none] (5) at (0.5, 0.25) {};
		\node [style=rn] (6) at (0.5, -0.25) {};
		\node [style=rn] (7) at (1, 0.25) {};
		\node [style=gn] (8) at (1, -0.25) {};
	\end{pgfonlayer}
	\begin{pgfonlayer}{edgelayer}
		\draw [bend right=45, looseness=1.00] (7) to (8);
		\draw [bend right=45, looseness=1.00] (8) to (7);
		\draw (7) to (8);
		\draw (0.center) to (1.center);
		\draw (3.center) to (4.center);
		\draw (5.center) to (6);
	\end{pgfonlayer}
\end{tikzpicture}\right)\right.&\left.\circ\interp{D}^{ZX\to\rec}\circ \left(\begin{tikzpicture}
	\begin{pgfonlayer}{nodelayer}
		\node [style=none] (0) at (-1, 0.25) {};
		\node [style=none] (1) at (-1, -0.25) {};
		\node [style=none] (2) at (-0.5, -0) {...};
		\node [style=none] (3) at (0, 0.25) {};
		\node [style=none] (4) at (0, -0.25) {};
		\node [style=none] (5) at (0.5, -0.25) {};
		\node [style=rn] (6) at (0.5, 0.25) {};
		\node [style=rn] (7) at (1, 0.25) {};
		\node [style=gn] (8) at (1, -0.25) {};
	\end{pgfonlayer}
	\begin{pgfonlayer}{edgelayer}
		\draw [bend right=45, looseness=1.00] (7) to (8);
		\draw [bend right=45, looseness=1.00] (8) to (7);
		\draw (7) to (8);
		\draw (0.center) to (1.center);
		\draw (3.center) to (4.center);
		\draw (5.center) to (6);
	\end{pgfonlayer}
\end{tikzpicture}\right)\right\rrbracket \\
&= \left(I\otimes \tredstate \right)\circ\left(A\otimes I_2 + B\otimes \iY\right)\circ\left(I\otimes \redstate\right)\\
&= A\otimes \left(\tredstate I_2\redstate\right) + B\otimes \left(\tredstate\iY\redstate\right) = A
\end{align*}
The proof is the same for the imaginary part, and for the other interpretation.
\end{proof}

This corollary is very helpful to show results on universality:

\begin{proposition}
The \reca is universal for real quantum transformations:
\[ \forall M\in \mathbb{R}^{2^n}\times\mathbb{R}^{2^m}, \exists D\in \rec, \interp{D} = M \]
\end{proposition}

\begin{proposition}
$\reind{\pi/4}$, the fragment of \rec-calculus that only uses angles multiples of
$\pi/4$ is approximately universal.
\end{proposition}

\bibliographystyle{eptcs}

\section{Appendix}

\textbf{Notation:} The boxes with $\pm\frac{\pi}{2}$ angles will be written
 in order to simplify some lemmas and proofs.

\subsection{Lemmas}

\begin{lemma}
\label{eq:box-0}
A box with angle $0$ is a mere wire.
\[%
\]
\end{proof}

\begin{lemma}
\label{eq:2scalars}
A node with no edge equals two ``bicolor'' scalars.
\[%
\]
\end{proof}

\begin{lemma}
\label{eq:hopf-law}
We have the Hopf Law:
\[%
\]
\end{lemma}
\begin{proof}
Using the rules \bo, \bt, \sth, \iv and lemma \ref{eq:2scalars}:
\[%
\]
\end{proof}

\begin{lemma}
\label{eq:generalised-bialgebra}
The rule \bt has a generalised version, derivable from \bt and \so.
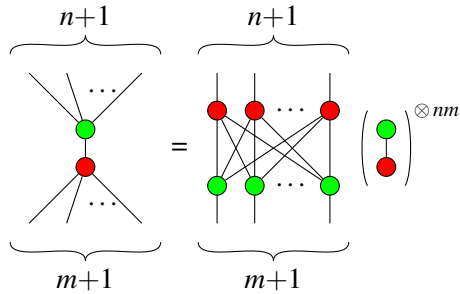
\[%
\]
\end{lemma}

\begin{remark}Notice that Lemma \ref{eq:generalised-bialgebra} has been proved, up to the scalars in \cite{euler-decomp}, and when $m=1$ in \cite{supplementarity}.\end{remark}

\begin{lemma}
\label{eq:ud-box}
The upside-down box $\alpha$ is the upright box with angle $-\alpha$.
\[%
\]
\end{proof}

\begin{lemma}
\label{eq:real-spider}
Two connected upright boxes merge with the sum of the two angles.
\[%
\]
\end{proof}

\begin{lemma}
\label{eq:pi-commutation}
The two hanging $\pi$ branches with inverted colors commute up to a scalar.
\[%
\]
\end{proof}

\begin{lemma}
\label{eq:pi-exclusion}
The $\pi$ hanging branch can be decomposed, making a ``$\pi/2$ boxes triangle'' appear.
\[%
\]
\end{proof}

\begin{lemma}
\label{eq:k2}
A $\pi$-branch can ``cross'' a real box, changing its orientation.
\[%
\]
\end{proof}

\begin{lemma}
\label{eq:red-state-green-pi}
A red state followed by a ``green'' $\pi$ hanging branch is equal to the mere red state.
\[%
\]
\end{proof}

\begin{lemma}
\label{eq:pi-involution}
Two hanging $\pi$ branches of the same color give the identity.
\[%
\]
\end{lemma}
\begin{proof}
Using \rh, \bo, the Hopf law \ref{eq:hopf-law} and \iv:
\[%
\]
\end{proof}

\begin{lemma}
\label{eq:exclusion1}
Using the $\pi$-branch decomposition, we can separate a real box from its main wire.
\[%
}\]
\end{proof}

\begin{lemma}
\label{eq:pi_2-loop}
A $\frac{\pi}{2}$-loop on a wire is just a wire, up to a scalar.
\[%
\]
\end{proof}

\begin{lemma}
\label{eq:exclusion2}
We can separate a box from its wire in another way than in lemma \ref{eq:exclusion1}.
\[%
}\]
\end{proof}

\begin{lemma}
\label{eq:2pi-box}
The $2\pi$-box is the identity, up to some scalar.
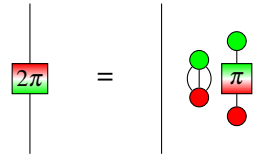
\[%
\]
\end{lemma}
\begin{proof}
First, we prove it on the green state, using \rsoo, \ref{eq:k2}, \rh, \ref{eq:pi-commutation} and \bo:
\[\fit{%
}\]
Now, in the general case, using \ref{eq:exclusion2}, the previous result and \ref{eq:pi_2-loop}:
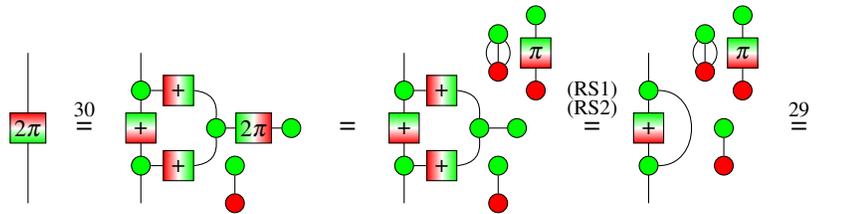
\[%
\]
\end{proof}

\begin{lemma}
\label{eq:-1^2}
Two copies of the previous scalar result in an empty diagram.
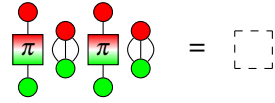
\[%
\]
\end{lemma}
\begin{proof}
Using the previous lemma (from right to left), \rso, \rsoo \rh and \iv:
\[%
\]
\end{proof}

\begin{proof}[Proof of Proposition \ref{prop:4_pi-periodical}]
\phantomsection\label{prf:4_pi-periodical}
Using the lemmas \ref{eq:real-spider}, \ref{eq:2pi-box} and \ref{eq:-1^2}:
\[%
}=-1$ so, by soundness of the \rec rules,
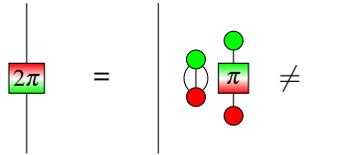
\[%
\]
\end{proof}

\begin{lemma}
\label{eq:supp-udbox}
One can flip the boxes on either side of the rule \rsup.
\[%
\]
Then, using \bt and the previous result:
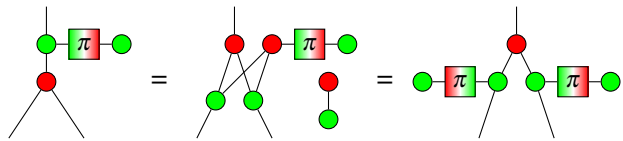
\[%
\]
\end{proof}

We have seen in section \ref{sec:pi_2-completeness} an interpretation that transforms a $\pi$ dot and a Hadamard yellow box into real boxes. Since everything works well with it, we would like to introduce the following notations in the \reca:
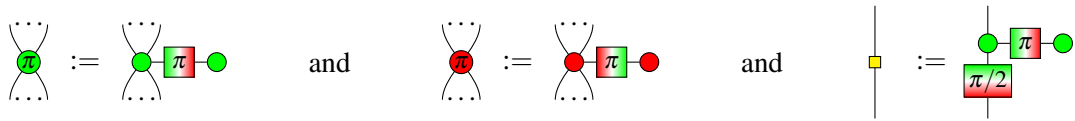
\[%
\]
With this notation, the same section shows that the \reca proves all the rules of the ZX$_r$.

Using lemma \ref{eq:k2}, one can easily show that:
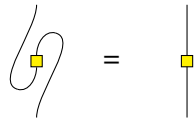
\[%
\]

\begin{lemma}
\label{eq:exclusion-hada}
The lemma \ref{eq:exclusion2} can be rewritten with Hadamard:
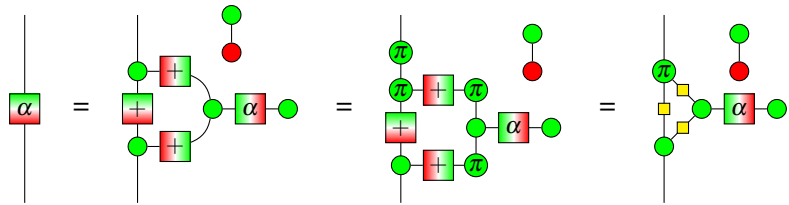
\[%
\]
\end{lemma}

\begin{lemma}
\label{eq:pi-box-decomp}
A real box $\pi$ is a green $\pi$-dot followed by a red one.
\[%
\]
\end{proof}

\subsection{Minimality}

\begin{proof}[Proof of Proposition \ref{prop:rs3-necessary}]
\hlabel{prf:rs3-necessary}
Let us consider the circular permutation $\sigma_n: k\mapsto (k+1)\mod n$, $(k\in \llbracket 0,n-1\rrbracket)$.\\
First, notice that: $\forall p\in \mathbb{Z}, ~~\sigma_n^p: k\mapsto k+p \mod n$.\\\\
We define a gate that has $n$ inputs and $n$ outputs: $U_{\sigma_n^p}$, which maps the $k$-th input to the $\sigma_n^p(k)$-th output.\\
We can notice that $\interp{U_{\sigma_n^p}}\circ\interp{U_{\sigma_n^q}}=\interp{U_{\sigma_n^p\circ\sigma_n^q}}=\interp{U_{\sigma_n^{p+q\mod n}}}$.\\
We can also notice that $\interp{R_Y(\alpha)}^{\otimes n}\circ\interp{U_{\sigma_n^p}} = \interp{U_{\sigma_n^p}}\circ\interp{R_Y(\alpha)}^{\otimes n}$\\
We now consider the following interpretation:
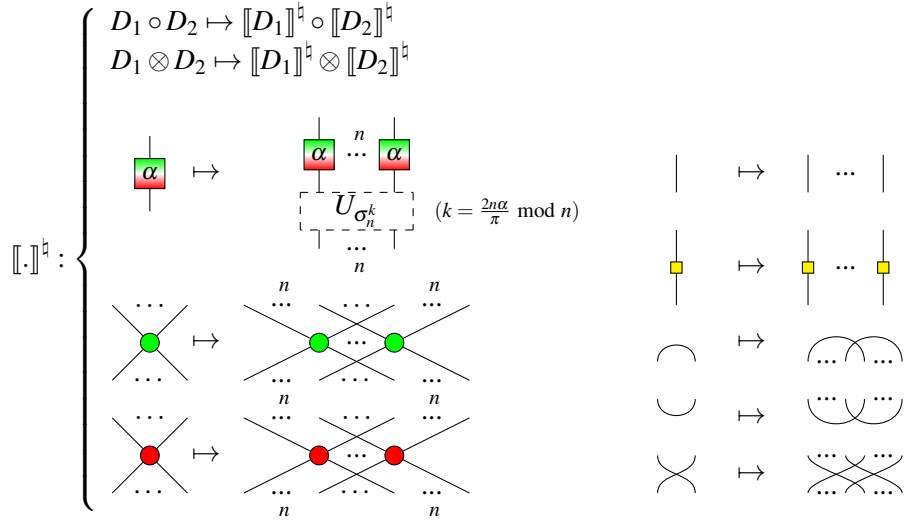
\[
\interp{.}^{\natural}: \left\lbrace 
\begin{array}{l}
D_1\circ D_2\mapsto \interp{D_1}^{\natural}\circ \interp{D_2}^{\natural}\\
D_1\otimes D_2\mapsto \interp{D_1}^{\natural}\otimes \interp{D_2}^{\natural}\\\\
\end{array}\right.
\]
Where $\interp{D_1\otimes D_2}^{\natural} = \interp{D_1}^{\natural}\otimes\interp{D_2}^{\natural}$ and $\interp{D_1\circ D_2}^{\natural} = \interp{D_1}^{\natural}\circ\interp{D_2}^{\natural}$ for any two diagrams $D_1$ and $D_2$.
\\
One can check that:
\[%
\]
\so, \st, \sth, \iv, \bo and \bt  obviously hold since no real box is used in these axioms.\\\\
\rsup holds: the interpretation only swaps identical hanging branches, which changes nothing.\\\\
\rh holds: $\sigma_n^0 = I^{\otimes n}$.\\\\
\rso holds:
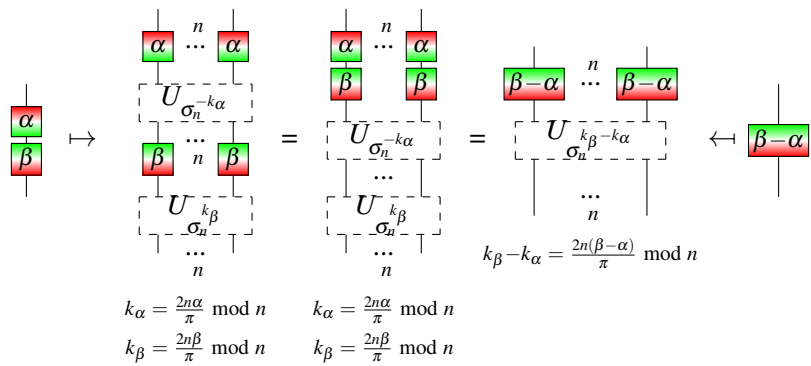
\[%
\]
\textbf{\rsth does not hold}: for $\alpha = \frac{\pi}{2n}\mod \frac{\pi}{2}$, i.e. $k=1$:\\Let us write to simplify: %
$}
\]
If \rsth were derivable from the other rules, its interpretation would hold, hence \rsth is necessary in any $\frac{\pi}{2n}$-fragment.
\end{proof}

\begin{proof}[Proof of Proposition \ref{prop:rh-necessary}]
\hlabel{prf:rh-necessary}
Let us consider the interpretation that maps any diagram $D:n\to m$ to the diagram $\interp{D}^{\sharp}:n\to m$ defined as:
\[ \interp{.}^{\sharp}:\left\lbrace \begin{array}{l}
D_1\circ D_2\mapsto \interp{D_1}^{\natural}\circ \interp{D_2}^{\natural}\\
D_1\otimes D_2\mapsto \interp{D_1}^{\natural}\otimes \interp{D_2}^{\natural}\\\\
\end{array} \right. \]

This interpretation obviously holds for \so, \st, \sth, \bo and \bt because no real box is involved in these rules, and all the rules hold when the colours are swapped and the boxes are flipped. \rso also holds, for no green or red dot appears here.

The rule \rsth holds. Using \rh, \rso and \rsth:
\[%
\]

Finally, the rule \rh does not hold. Indeed for dots of arity 1:
\[%
\]

\end{proof}

\subsection{The completeness of the \frag{2}}

\subsubsection{The real stabiliser ZX-Calculus}

\label{sec:real-stab-zx}

The real stabiliser ZX-Calculus -- its syntax and its set of axioms -- is defined in definition \ref{def:zxr}.

From these rules, we can derive \cite{pivoting}:
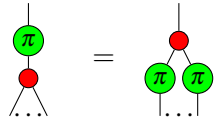
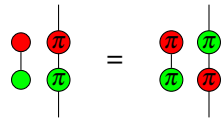
\begin{lemma}
\label{eq:zx-k1}
\[%
\]
\end{lemma}
\begin{lemma}
\label{eq:zx-zo-absorbing-scalar}
The dot $\pi$ has an absorbing property for any scalar i.e. any diagram with $0$ input and $0$ output.
\[%
\]
\end{lemma}

\subsubsection{$\interp{.}^{\repiot\to ZX_r}$ preserves the rules:}

\label{prf:rules-preserved}

The rules \so, \st, \sth, \iv, \bo, \bt obviously hold since no real box appears in them. \rso also holds, quite immediately.

\rsth holds thanks to the pivoting \cite{pivoting}. Using \zxrh, \zxrbt, \zxrhl, \ref{eq:zx-k1}:
\[ \fit{%
} \]
then, using \zxrh, the previous result, lemmas \ref{eq:zx-k2} and \ref{eq:zx-k1}, and \zxrso:
\[ \fit{%
} \]
and the result for $k\pi/2$ is obtained by applying $k$ times the results above.

\rh holds. Using \zxrh, \ref{eq:zx-k1} and the $2\pi$-periodicity of green dots:
\[ %
 \]

\subsubsection{$\interp{.}^{ZX_r \to\repiot}$ preserves the rules:}

First, the rules \zxrst, \zxrsth, \zxriv, \zxrbo and \zxrbt obviously hold because no yellow box and no angle are involved.

\zxrso obviously holds when either $\alpha$ or $\beta$ is null. When both are $\pi$, then the lemma \ref{eq:pi-involution} is used to show \zxrso holds

\zxrhl holds.
Indeed, using \rso and \ref{eq:pi_2-loop}:
\[%
}\]

\subsubsection{$\repiot \vdash D = \interp{\interp{D}^{\repiot \to ZX_r}}^{ZX_r \to \repiot}$}
\label{sec:pi_2-double-interpretation}
For any \repiot-diagram $D$:
\[ \repiot \vdash D = \interp{\interp{D}^{\repiot \to ZX_r}}^{ZX_r \to \repiot} \]
Indeed, using lemmas \ref{eq:pi-involution} and \ref{eq:real-spider}:
\[%
\]
The reasoning is the same for the upside-down box, and otherwise, the composition of the interpretations is the identity.

\subsection{The ZX-Calculus}
\label{sec:zx-calculus}

The general ZX-Calculus -- its syntax and its set of axioms -- is presented in section \ref{sec:zx-def}. From these rules can be derived the lemmas:
\begin{lemma}
\label{eq:zx-mult-glob-ph}
\[%
\]
\end{lemma}

\begin{proof}[Proof of Lemma \ref{lem:commutation-on-control-wire}]
\phantomsection\label{prf:commutation-on-control-wire}
Using \ref{eq:exclusion-hada}, \ref{eq:rs3-rewritten}, \ref{eq:k2} and \ref{eq:swap-hadamard}:
\[ {%
}\]
\end{proof}

\begin{proof}[Proof of Proposition \ref{prop:y-to-zx-interpretation}]
\hlabel{prf:y-to-zx-interpretation}
\so, \st, \sth, \bo and \bt obviously hold. (ZO) also holds, the demonstration is the same as for $\interp{.}^{\repiot\to ZX_r}$.\\
\rso holds. Using \zxkt and lemma \ref{eq:zx-mult-glob-ph}:
\[ %
 \]
\rsth holds. Using lemma \ref{eq:zx-eu'}, \zxso, \zxh, \zxkt, \zxbt:
\[ \fit{%
} \]
\end{proof}

\begin{proof}[Proof of Proposition \ref{prop:zx-to-y-interpretation}]
\hlabel{prf:zx-to-y-interpretation}
First notice that:
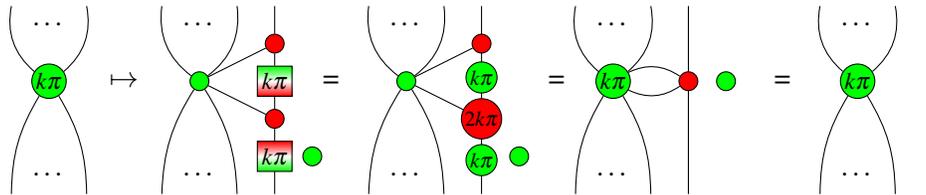
\[ %
\]
The result is the same with a red dot. Hence, all the rules that only display red and green dots of angles $0$ -- \zxst, \zxsth, \zxbo, \zxbt -- are obviously preserved.\\\zxh holds:
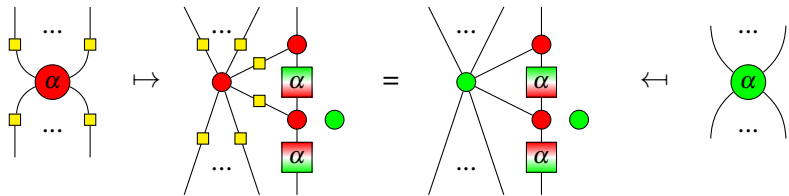
\[ %
\]
\zxso hold. Using lemmas \ref{eq:rs3-rewritten}, \ref{eq:real-spider} and \ref{eq:hopf-law}:
\[ %
\]
Then adding inputs and outputs to the green dots does not change anything.\\
\zxkt holds. Using \ref{eq:k1}, \ref{eq:box-0}, \rso, \ref{eq:k2} and \ref{eq:pi-involution}:
\[ %
\]
\zxeu holds. First notice that, using \zxbt, \zxh and \zxrhl:
\[ %
\]
Then, using the decomposition of $\pi/2$ boxes, \ref{eq:rs3-rewritten}, \zxbt, the expression of \rsth in the ZX, \ref{eq:hopf-law} and \ref{eq:k1}:
\[ %
\]
\zxe holds. Using \ref{eq:rs3-rewritten}, \zxrh, \bt, \zxrhl, \ref{eq:hopf-law}, \ref{eq:pi-involution}, \st, the Hadamard decomposition, \rsup, and \rh:
\[%
\]
\end{proof}

\begin{proof}[Proof of Proposition \ref{prop:zx2y-decomp}]
\hlabel{prf:zx2y-decomp}
By induction on the diagram:\\
$\bullet$ \textbf{Base Cases:} Showing the result for a green or red dot with only one wire is just a bit of computation. Then, using \zxso, the result can be extended to a green/red dot of any arity. The result is obvious for all other generators.\\
$\bullet$ \textbf{Sequential Composition:} Let two diagrams $D_1$, $D_2$, and four real matrices $A_1$, $B_1$, $A_2$, $B_2$ such that:
\[ \interp{D_1}=A_1+iB_1 \qquad and\qquad \interp{D_2}=A_2+iB_2 \]
We suppose that the result is true for $D_1$ and $D_2$:
\[
\interp{\interp{D_1}^{ZX\to\rec}} = A_1\otimes I_2 + B_1\otimes \iY\qquad and\qquad
\interp{\interp{D_2}^{ZX\to\rec}} = A_2\otimes I_2 + B_2\otimes \iY
\]
On the one hand:
\begin{align*}
\interp{\interp{D_2\circ D_1}^{ZX\to\rec}} &= \interp{\interp{D_2}^{ZX\to\rec}}\circ\interp{\interp{D_1}^{ZX\to\rec}}\\
&= \left(A_2\otimes I_2 + B_2\otimes \iY\right)\circ \left(A_1\otimes I_2 + B_1\otimes \iY\right)\\
&= \left((A_2\circ A_1)-(B_2\circ B_1)\right)\otimes I_2 + \left((A_2\circ B_1)+(B_2\circ A_1)\right)\otimes \iY
\end{align*}
On the other hand:
\begin{align*}
\interp{D_2\circ D_1} &= (A_2+iB_2)\circ(A_1+iB_1) \\
&= (A_2\circ A_1)-(B_2\circ B_1) + i(A_2\circ B_1)+(B_2\circ A_1)
\end{align*}
And thus:
\[ \interp{\interp{D_2\circ D_1}^{ZX\to\rec}} = \Re\left(\interp{D_2\circ D_1}\right)\otimes I_2 + \Im\left(\interp{D_2\circ D_1}\right)\otimes \iY \]
$\bullet$ \textbf{Spatial Composition:} With the same diagrams and matrices (we still assume that the result is true for $D_1$ and $D_2$).\\
On the one hand ($m$ being the number of inputs of $D_2$ and $D_1$ having $n$ inputs and $n'$ outputs):
\begin{align*}
\interp{\interp{D_1\otimes D_2}^{ZX\to\rec}} &= \left( I_2^{\otimes n'}\otimes\interp{\interp{D_2}^{ZX\to\rec}}\right)\circ\interp{\nmswapI{m}{n'}}\\
&\qquad\circ \left(I_2^{\otimes m}\otimes\interp{\interp{D_1}^{ZX\to\rec}}\right)\circ\interp{\nmswapI{n}{m}}\\
&= \left( I_2^{\otimes n'}\otimes A_2 \otimes I_2 + I_2^{\otimes n'}\otimes B_2\otimes \iY\right)\circ\interp{\nmswapI{m}{n'}}\\
&\qquad\circ \left(I_2^{\otimes m}\otimes A_1 \otimes I_2 + I^{\otimes m}\otimes B_1\otimes \iY\right)\circ\interp{\nmswapI{n}{m}}\\
&= \left( I_2^{\otimes n'}\otimes A_2 \otimes I_2 + I_2^{\otimes n'}\otimes B_2\otimes \iY\right)\\
&\qquad\circ\left(A_1\otimes I_2^{\otimes m} \otimes I_2 + B_1\otimes I_2^{\otimes m}\otimes \iY\right)\\
&= ((A_1\otimes A_2)-(B_1\otimes B_2)) \otimes I_2 + ((A_1\otimes B_2)+(B_1\otimes A_2)) \otimes \iY
\end{align*}
On the other hand:
\[ \interp{D_1\otimes D_2} = (A_1\otimes A_2)-(B_1\otimes B_2) + i((A_1\otimes B_2)+(B_1\otimes A_2)) \]
Thus:
\[ \interp{\interp{D_1\otimes D_2}^{ZX\to\rec}} = \Re\left(\interp{D_1\otimes D_2}\right)\otimes I_2 + \Im\left(\interp{D_1\otimes D_2}\right)\otimes \iY \]
\end{proof}

\begin{proof}[Proof of Proposition \ref{prop:inverses}]
\phantomsection\label{prf:inverses}
Let us prove that 
\[\forall D\in\rec,~~ \rec\vdash \interp{\interp{D}^{\rec\to ZX}}^{ZX\to\rec}\circ\left(
\right) \]
then the result directly follows thanks to \ref{eq:bicolor-alpha-0} and \zxiv.

This is obvious for every generator but \scalebox{0.7}{%
\end{align*}
Showing the result for a spatial composition is a bit of computation, and for the sequential composition, it is obvious. Then, by induction, we prove the result for any diagram.
\end{proof}

\end{document}